%% file: it.tex
\documentclass[12pt, journal, onecolumn]{IEEEtran}
\usepackage{cite}
\usepackage[cmex10,intlimits]{amsmath}
\usepackage{amsfonts}
\usepackage{graphicx}
\usepackage{epsfig}
\usepackage{subfig}
\usepackage{amssymb}
\usepackage{amsthm}
\usepackage[body={6.5in,9in},top=1in, left=1in]{geometry}
\usepackage{mathrsfs}
\usepackage{pstricks}
\include{consts}

\newtheorem{theorem}{Theorem}
\newtheorem{proposition}{Proposition}
\newtheorem{lemma}{Lemma}

\title{Distributed Beamforming in Wireless Multiuser Relay-Interference Networks with Quantized Feedback}
\begin{document}
\renewcommand{\qedsymbol}{$\blacksquare$}
\author{Erdem Koyuncu and Hamid Jafarkhani, \emph{Fellow}, IEEE}
\renewcommand{\thefootnote}{}
\footnotetext{
This work was presented in part at IEEE Global Communications Conference (GLOBECOM), Nov. 2009.}
\footnotetext{The authors are with the Center for Pervasive Communications and Computing, University of California, Irvine, Irvine CA 92697-2625 USA. Email: \{ekoyuncu, hamidj\}@uci.edu.}
\renewcommand{\thefootnote}{\arabic{footnote}}
\maketitle

\begin{abstract}
We study quantized beamforming in wireless amplify-and-forward relay-interference networks with any number of transmitters, relays, and receivers. We design the quantizer of the channel state information to minimize the probability that at least one receiver incorrectly decodes its desired symbol(s). Correspondingly, we introduce a generalized diversity measure that encapsulates the conventional one as the \textit{first-order} diversity. Additionally, it incorporates the \textit{second-order} diversity, which is concerned with the transmitter power dependent logarithmic terms that appear in the error rate expression. First, we show that, regardless of the quantizer and the amount of feedback that is used, the relay-interference network suffers a second-order diversity loss compared to interference-free networks. Then, two different quantization schemes are studied: First, using a global quantizer, we show that a simple relay selection scheme can achieve maximal diversity. Then, using the localization method, we construct both fixed-length and variable-length local (distributed) quantizers (fLQs and vLQs). Our fLQs achieve maximal first-order diversity, whereas our vLQs achieve maximal diversity. Moreover, we show that all the promised diversity and array gains can be obtained with arbitrarily low feedback rates when the transmitter powers are sufficiently large. Finally, we confirm our analytical findings through simulations.
\end{abstract}
\section{Introduction}
While it has been demonstrated in several studies that cooperation can greatly improve the performance and reliability of wireless network communications\cite{nabar1, bolcksei1, laneman2, laneman1, sendonaris1}, interference still remains to be a fundamental issue in cooperative network design. Most of the previous work on cooperative networks relies on orthogonal channel allocation so that different transmitters do not interfere with each other. However, allocating orthogonal channels for each user may not be desirable due to time and bandwidth limitations\cite{oyman1, li2}. In such cases, one should explore effective ways to deal with interference while preserving cooperative diversity gains.

Multiple antenna interference cancelation techniques are very effective when dealing with interference in cooperative networks\cite{jing8}. They offer reasonable performance with low decoding complexity. In this work, we consider a different approach. To be able to study the ultimate performance limits, we do not put any restrictions on our decoders. We would like to design a cooperation scheme that achieves maximal diversity benefits, and thus provides high reliability, even in the presence of multiuser interference.

For networks with a single transmitter-receiver pair and no interference, network beamforming using amplify-and-forward (AF) relays has shown to achieve the maximal spatial diversity\cite{larsson1, jing1}. However, the optimal beamforming policy requires one or two real numbers to be broadcasted from the receiver to the relays. Using distributed beamforming with quantized instantaneous channel state information (CSI), it is possible to obtain both maximal diversity, as well as high array gain with only a few feedback bits from the receiver\cite{koyuncu1, zhao2, ahmed1}. A special case of quantized feedback for cooperative networks is the relay selection scheme\cite{anghel1,hasna1,jing6}. It has been formally shown in \cite{koyuncu1} that, for a network with $R$ parallel relays, the relay selection scheme provides the maximum diversity $R$.

Quantized feedback schemes have also been studied for non-cooperative multiuser interference networks. In \cite{jindal1}, the author considers zero-forcing beamforming with finite rate feedback in multiple-input multiple-output (MIMO) broadcast channels. Interference alignment for multiuser interference networks with limited feedback has been studied in \cite{thukral1}. Unlike what we shall study in this work, where we seek to optimize the reliability of the system in terms of the diversity gain, the goal of the above two papers was to optimize the data transmission rate in terms of the multiplexing gain. A common conclusion that we can infer from both studies is that, in order to achieve the same multiplexing gain as a system with perfect CSI, the feedback rate should be increased at least logarithmically with the transmitter power; any constant feedback rate results in a complete loss of multiplexing gain. This is unlike point-to-point systems where feedback is not even necessary to achieve the maximal multiplexing gain\cite{jindal1}, and a few feedback bits is usually sufficient to transmit with rates that are close to the one with perfect CSI\cite{roh1}. The feedback requirements of interference networks appears to be considerably higher than that of interference-free networks.

What are the feedback requirements if instead we would like to ensure maximal reliability in the presence of interference? One goal of this paper is to answer this question for cooperative networks with $K$ transmitters, $L$ receivers, and $R$ parallel AF relays. We assume that each transmitter and each relay has its own short term power constraint. The transmitters do not have any CSI. Each receiver knows its own receiving channels and the channels from the transmitters to the relays. Each relay only knows the magnitudes of its own receiving channels. Each relay and each receiver also has partial CSI provided by feedback. The feedback information represents a quantized beamforming vector. In that sense, this paper is also a generalization of single-user quantized network beamforming \cite{koyuncu1} to multiuser interference networks. On the other hand, such a generalization is quite challenging because of the distributed nature of the network. Let us now describe some of these challenges and our approaches to address them.

In interference networks, the relays amplify both noise and interference, which results in completely different problem formulations and solutions. Second, there are multiple receivers that have different optimal beam directions. As a result, it is difficult to design a scheme that can provide a reasonable performance to all the users.

Another difficulty is related to acquiring feedback information from several separated receivers. The optimal beamforming policy requires the full CSI of the interference network. In practice however, none of the receivers can obtain such information via training methods. We thus consider two different quantization schemes: In the first scheme, the feedback information is provided by a global quantizer (GQ) that knows the entire CSI. We use this hypothetical quantizer to analyze the performance limits of network beamforming in the presence of interference. In the more practical second scheme, we use distributed local quantizer (LQ) encoders at each receiver. Each receiver can access only a part of the CSI, and provides its own feedback information for the relays and the other receivers.

In \cite{koyuncu2}, we introduced a general systematic LQ design method, called localization, in which one synthesizes an LQ out of an existing GQ using high-rate scalar quantization combined with entropy coding. In the same work, we described an application of the method to MIMO broadcast channels. In this work, we apply it to design LQs for our network model. Therefore, our GQ has another important purpose other than the one we have previously mentioned: It  will also serve as the basis of our LQs.

We would also like to note that the LQ design in this paper distinguishes itself from the one in \cite{koyuncu2} in several ways, even though the underlying localization method will be the same. First, we need to consider a totally different and much more complicated distortion function. Second, the high-rate scalar quantizers, that form the crucial part of the method, should be designed accordingly. Third, the performance analysis of the resulting LQs is thus different and more complicated. As a result, in this work, we will only analyze the performance of localization for a particular class of GQs that are based on relay selection.

Our performance measure is what we call the \textit{network error rate} (NER). Given a fixed channel state, it is the probability that at least one user incorrectly decodes its desired symbol(s). In that sense, any receiver can be interested in the symbols transmitted by any subset of transmitters.

We use a generalized diversity measure to characterize the asymptotic behavior of the NER as the transmitter powers grow to infinity. In what follows, we describe this measure together with its motivations: Suppose that a wireless communication system achieves an error rate of $C(P^{\alpha} \log^{\beta}P)^{-1}$, where $P$ is the transmitter power constraint and $C$ is a constant that is independent of $P$. Then, we call $\alpha$ and $\beta$, the first-order and the second-order diversity gains, respectively, and say that the scheme achieves diversity $(\alpha, \beta)$. Such a definition of diversity is more precise than the traditional one as we demonstrate by an example: For two hypothetical communication systems with diversity gains $(\alpha, \beta_1)$, and $(\alpha, \beta_2)$, where $\alpha\geq 1$ and $\infty>\beta_1 > \beta_2>-\infty$, the former always outperforms the latter for all $P$ sufficiently large. On the other hand, the traditional definition, according to which the diversity gain is $\alpha$ for both systems, fails to distinguish between the asymptotic performance of the two.

The main contributions of this paper can be summarized as follows: First, we show that, regardless of the quantizer and the amount of feedback that is used, the maximal achievable diversity of our network model is $(R, -R)$ when $K>1$, whereas it is $(R, 0)$ when $K=1$.\footnote{The case $K=1$ corresponds to a relay-broadcast network that does not suffer any multiuser interference. Even though our main goal in this paper is to analyze interference networks, we present the extension of our results to broadcast networks, so as to demonstrate the detrimental effects of interference in a comparative manner.} In other words, the relay-interference network suffers from a second-order diversity loss compared to an interference-free network that can achieve diversity $(R, 0)$ with $K=L=1$\cite{koyuncu1}. Then, we construct a relay-selection based fixed-length GQ (fGQ) that can achieve maximal diversity for any $K$. Next, using our fGQ and the localization method, we design both fixed-length and variable-length LQs (fLQs and vLQs). Our fLQs can achieve diversity $(R, -2R)$ when $K>1$, and diversity $(R, -R)$ when $K=1$, using $R$ feedback bits per receiver. They show that it is possible to achieve very high reliability using a fixed number of feedback bits.
On the other hand, our vLQs can achieve maximal diversity gain for any $K$. Moreover, the feedback rate they require decays to zero as the transmitter powers grow to infinity. Therefore, they provide a very fortunate answer to the question that we have posed earlier: In a relay-interference network, it is possible to achieve maximal reliability using arbitrarily low feedback rates per receiver, when the transmitter powers are sufficiently large. Another desirable property of our vLQs is the fact that the array gain they provide can be made arbitrarily close to the one provided by the fGQ.

The rest of the paper is organized as follows: In Section \ref{secIntro}, we introduce our network model, performance and diversity measures, and problem definition. In Section \ref{alowbound}, we show that the maximal diversity of our network model is $(R, -R)$. In Sections \ref{secglobal} and \ref{seclocaldddd}, we introduce our GQ and LQ designs, respectively. Numerical results are provided in Section \ref{secsims}. In Section \ref{secconcs}, we draw our major conclusions. An upper bound on the probability density function (PDF) and the cumulative distribution function (CDF) of a frequently used random variable (RV) is provided in Appendix \ref{dendededen}. Some other technical proofs are provided in Appendices \ref{proofoftheorem1} through \ref{proofoftheorem4}.

\textit{Notation: }For a logical statement $\mathtt{S}$, ``$\mathtt{S}$ is true for $x$ sufficiently large'' means that there exists $x_0 <\infty$ such that for all $x\geq x_0$, $\mathtt{S}$ is true. $\|\cdot\|$ indicates the 2-norm, $\|\cdot\|_{\infty}$ is the infinite norm, $\langle\cdot|\cdot\rangle$ is the inner product. $\mathbb{C}$, $\mathbb{R}$ and $\mathbb{Z}^+$ represent the sets of complex numbers, real numbers, and positive integers, respectively. $\det(\mathbf{A})$ is the determinant of a square matrix $\mathbf{A}$. $\mathbf{A}^T$, $\mathbf{A}^H$ denote the transpose and the Hermitian transpose of $\mathbf{A}$, respectively. $\mathtt{P}$ represents the probability. $f_X(\cdot)$ is the PDF, and $F_X(\cdot)$ is the CDF of an RV $X$. $\mathtt{E}[X]$ is the expected value of $X$. $X\sim
\Gamma(k,\theta)$ means that $X$ is a Gamma RV with $f_X(x)= \frac{x^{k-1}e^{-x/\theta}}{\theta^k \Gamma(k)}$ for $x> 0$ and $f_X(x)=0$ for $x \leq 0$, $k,\theta > 0$. For any sets $\mathcal{A}$ and $\mathcal{B}$, $\mathcal{A}-\mathcal{B}$ is the set of elements in $\mathcal{A}$, but not in $\mathcal{B}$. $|\mathcal{A}|$ is the cardinality of $\mathcal{A}$. $\mathcal{A}^r = \left\{(a_1,\ldots,a_r):a_1,\ldots, a_r\in\mathcal{A}\right\}$, $r\in\mathbb{Z}^+$, is the cartesian power. $\gamma_e = 0.577...$ is the Euler-Mascheroni constant, $e = \exp(1)$, and $\emptyset$ is the empty set. For a real-valued function $f:\mathcal{C} \rightarrow \mathbb{R}$ with $\mathcal{C}\subset\mathbb{C}^K$, let
$\mathcal{M} \triangleq \{\mathbf{x}:\mathbf{x}\in\mathcal{C},\,f(\mathbf{x}) = \max_{\mathbf{x}'\in\mathcal{C}} f(\mathbf{x'})\}$. Then, $\arg\max_{\mathbf{x}\in\mathcal{C}}f(\mathbf{x})$ is the unique vector $\mathbf{x}^*$ with the property that $\mathbf{x}^* \prec \mathbf{x},\,\forall \mathbf{x}\in\mathcal{M}$, and ``$\prec$'' represents some partial ordering (e.g. lexicographical ordering) of complex vectors. We define $\arg\min(\cdot)$ in a similar manner. Finally, $\log(\cdot)$ is the natural logarithm, $\log_2(\cdot)$ is the logarithm to base $2$, $\cosh(\cdot)$ is the hyperbolic cosine, $\mathrm{Q}(\cdot)$ is the Gaussian tail function, $\Gamma(\cdot)$ is the gamma function, $E_1(x) \triangleq \int_1^{\infty} e^{-1}e^{-xt}\mathrm{d}t$ is the exponential integral, and $K_\nu(\cdot)$ is the modified Bessel function of the second kind of order $\nu$.
\section{Network Model and Problem Statement}
\label{secIntro}
\subsection{System Model}
The block diagram of the system is shown in Fig. \ref{blockdiagram}. We
have a relay network with $K$ transmitters, $L$ receivers, and $R$ parallel relays. The cases $K=1$ and $K>1$ correspond to a relay-broadcast network and a relay-interference network, respectively. We assume that there is no direct link between the transmitters and the receivers.
\begin{figure}[h]
\centering
\scalebox{1.5}{\epsfig{file = 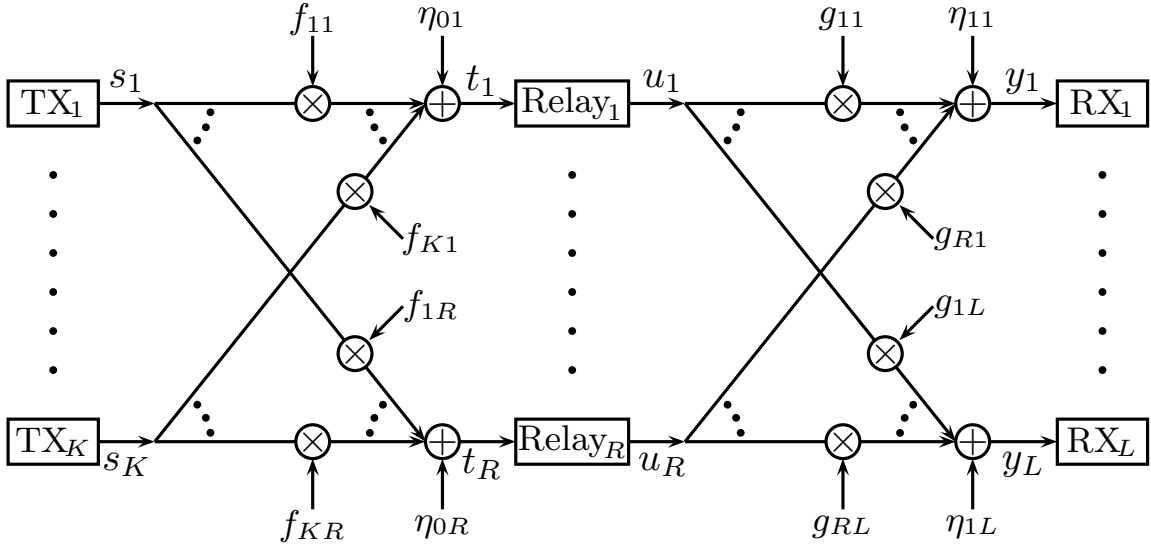, bbllx = 111pt, bblly = 563pt, bburx = 399pt, bbury = 700pt}}
\caption{System block diagram. In the figure, $\mathrm{TX}_k$, $\mathrm{Relay}_r$, and $\mathrm{RX}_{\ell}$ stand for the $k$th transmitter, $r$th relay, and the $\ell$th receiver, respectively.}
\label{blockdiagram}
\end{figure}

Denote the channel from the $k$th transmitter to the $r$th relay by
$f_{kr}$ and the channel from the $r$th relay to the $\ell$th receiver by $g_{r\ell}$.
Let $\mathbf{h} = (f_{11},\ldots,f_{KR},g_{11},\ldots,g_{RL})$ denote the channel state of the entire network.
We assume that the entries of $\mathbf{h}$ are independent and distributed as
$f_{kr} \sim \mathcal{CN}(0,\sigma_{f_{kr}}^2)$, $g_{r\ell} \sim \mathcal{CN}(0,\sigma_{g_{r\ell}}^2)$ with finite variances $\sigma_{f_{kr}},\sigma_{g_{r\ell}} < \infty,\,\forall r,k,\ell$.
For brevity, let $\mathbf{g}_\ell \triangleq (g_{1\ell},\ldots,g_{R\ell})$, which denotes all the channels from the relays to the $\ell$th receiver.

Only the short-term power constraint is considered, which means
that for every symbol transmission, the average power levels used at the $k$th transmitter
and the $r$th relay are no larger than $P_{S_k}$ and $P_{R_r}$, respectively.

We assume a quasi-static channel model; the channel
realizations vary independently from one channel state to another, while within each channel state the channels remain constant. We assume that
the $\ell$th receiver knows $\mathbf{g}_{\ell}$ and each relay knows the magnitudes of its own receiving channels, i.e. the $r$th relay knows $|f_{kr}|,\,k=1,\ldots,K$. Some possible procedures to reveal the channel states to the receivers can be found in \cite{ahmed1, koyuncu1}.
For completeness, we give an outline of one possible way: The $\ell$th destination can acquire the knowledge of $g_{r\ell}$ by training from the $r$th relay. The $r$th relay can acquire the knowledge of $|f_{kr}|$ using training sequences from the $k$th source. It can also amplify and forward its received training signal from the source to the destination, so that the destination can estimate the product of $f_{kr}$ and $g_{r\ell}$. As $g_{r\ell}$ is known by the destination, $f_{kr}$ can be estimated.

Each relay and each receiver also has partial CSI provided by feedback. In this paper, we consider two different feedback schemes, namely the global and local quantization schemes.
\subsection{Global Quantization}
\label{gquantizer}
Our global quantizer $\mathtt{GQ}$ is defined by a global encoder and a global decoder, as described in Fig. \ref{globop}. The global encoder consists of two parts. For each channel state, first, a GQ encoder $\mathtt{QGE}:\mathbb{C}^{R(K+L)}\rightarrow\mathcal{I}^\mathtt{G}$ maps the channel realization $\mathbf{h}$ to an index in $\mathcal{I}^\mathtt{G}\triangleq\{1,\ldots,|\mathcal{I}^\mathtt{G}|\}$, the index set of the codebook elements. Then, a lossless global compressor $\mathtt{GQC}:\mathcal{I}^\mathtt{G}\rightarrow \mathcal{J}^\mathtt{G}$ maps this index to a binary description.
\begin{figure}[h]
\centering
\scalebox{1}{\epsfig{file = 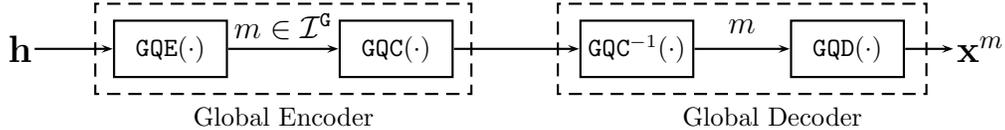, bbllx = 144pt, bblly = 432pt, bburx = 519pt, bbury = 480pt}}
\caption{Global quantizer operation.}
\label{globop}
\end{figure}

Let $\mathfrak{l}(j)$ denote the length of a binary description $j$. We call $\mathtt{GQ}$ a \textit{fixed-length} GQ (fGQ) if $\mathfrak{l}(j) = \lceil \log_2 |\mathcal{I}^\mathtt{G}| \rceil,\,\forall j\in\mathcal{J}^\mathtt{G}$. Otherwise, we call $\mathtt{GQ}$ a \textit{variable-length} GQ (vGQ).

In either case, the global encoder feeds back $\mathtt{GQC}(\mathtt{GQE}(\mathbf{h}))$, using $\mathfrak{l}(\mathtt{GQC}(\mathtt{GQE}(\mathbf{h})))$ bits. The feedback bits are received by the global decoders without any errors or delays.

There is a unique global decoder at each relay and each receiver, which comprises of the complementary parts to the global encoder: A lossless decompressor and a quantizer decoder. First the decompressor $\mathtt{GQC}^{-1}:\mathcal{J}^\mathtt{G}\rightarrow \mathcal{I}^\mathtt{G}$ reconstructs the quantization index from the received binary description. It is followed by the quantizer decoder $\mathtt{GQD}:\mathcal{I}^\mathtt{G}\rightarrow\mathcal{C}^\mathtt{G}$ which maps the quantization index to a codebook element. The codebook $\mathcal{C}^\mathtt{G}$ has $|\mathcal{I}^\mathtt{G}|$ elements,  $\mathcal{C}^\mathtt{G} = \{\mathbf{x}_1, \ldots, \mathbf{x}_{|\mathcal{I}^\mathtt{G}|}\}$. Without loss of generality, for
$\mathtt{GQE}(\mathbf{h})=m$, we set $\mathtt{GQD}(m)=\mathbf{x}_m\in\cbg$. For the rest of this paper, we will use the well-known notation $\quantg(\mathbf{h})\triangleq(\mathtt{GQD}\circ\mathtt{GQC}^{-1}\circ \mathtt{GQC} \circ \mathtt{GQE})(\mathbf{h})=(\mathtt{GQD} \circ \mathtt{GQE})(\mathbf{h})$. Therefore, $\mathtt{GQ}:\mathbb{C}^{R(K+L)}\rightarrow\mathcal{C}^{\mathtt{G}}$, and $\mathtt{GQ}(\mathbf{h})=\mathbf{x}$, for some $\mathbf{x}\in\mathcal{C}^\mathtt{G}$.

In the most general case, the $r$th relay may make use of the side information $|f_{kr}|$ in the process of decoding the feedback information. However, in order to keep the relay operation as simple as possible, we do not consider such a scenario in this paper.
\subsection{Local Quantization}
\label{locquantexplanation}
We define our local quantizer $\mathtt{LQ}$ by $L$ local encoders, with the $\ell$th encoder at the $\ell$th receiver, and a unique local decoder at each receiver and relay, as described in Fig. \ref{locop}. The $\ell$th local encoder comprises of two parts: An LQ encoder $\mathtt{LQE}_{\ell}:\mathbb{C}^{R(K+1)}\rightarrow
\mathcal{I}_\ell^\mathtt{L}$ and a lossless local compressor $\mathtt{LQC}_{\ell}:\mathcal{I}_\ell^\mathtt{L}\rightarrow \mathcal{J}_\ell^\mathtt{L}$. Note that the domain of each LQ encoder is different from the domain of the GQ encoder. For the $\ell$th encoder, the domain corresponds to the channel states from the transmitters to the relays and from the relays to the $\ell$th receiver, represented by the concatenation vector $[\mathbf{f},\mathbf{g}_\ell]$.
\begin{figure}[h]
\centering
\scalebox{1}{\epsfig{file = 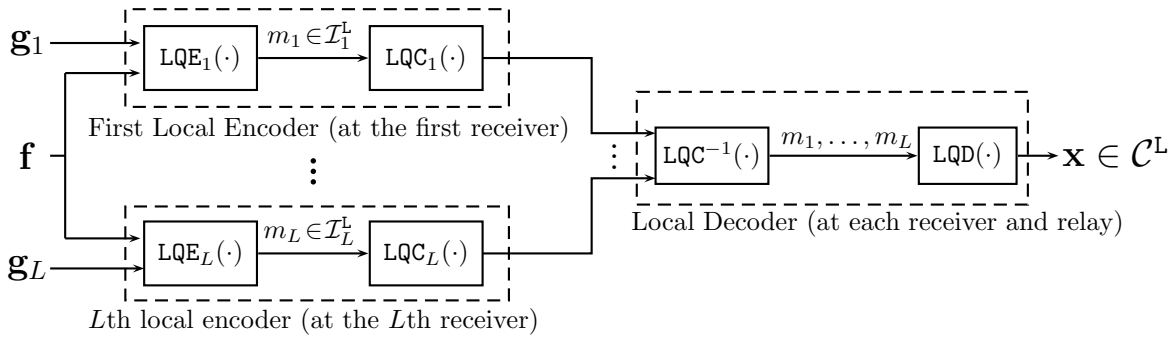, bbllx = 131pt, bblly = 469pt, bburx = 570pt, bbury = 595pt}}
\caption{Local quantizer operation.}
\label{locop}
\end{figure}

The $\ell$th receiver feeds back $\mathtt{LQC}_{\ell}(\mathtt{LQE}_{\ell}([\mathbf{f},\mathbf{g}_\ell]))$, using $\mathfrak{l}(\mathtt{LQC}_{\ell}(\mathtt{LQE}_{\ell}([\mathbf{f},\mathbf{g}_\ell])))$ bits. We call $\mathtt{LQ}$ an fLQ if, $\mathfrak{l}(j) = \lceil \log_2 |\mathcal{J}_\ell^\mathtt{L}| \rceil,\,\forall j\in\mathcal{J}_\ell^\mathtt{L},\,\forall \ell$. Otherwise, we call it a vLQ. For the latter case, the feedback rate of the $\ell$th receiver can be expressed as $\mathtt{R}_{\ell}(\mathtt{LQ}) \triangleq \mathtt{E}[\mathfrak{l}(\mathtt{LQC}_{\ell}(\mathtt{LQE}_{\ell}([\mathbf{f},\mathbf{g}_\ell])))]$.

After all the $L$ feedback messages are exchanged between the receivers and the relays, each of them decodes the feedback bits using the local decoder. The local decoder is the composition of a decompressor $\mathtt{LQC}^{-1}:\prod_{\ell}\mathcal{J}_\ell^\mathtt{L} \rightarrow \prod_{\ell}\mathcal{I}_\ell^\mathtt{L}$ and a quantizer decoder $\mathtt{LQD}:\prod_{\ell}\mathcal{I}_\ell^\mathtt{L} \rightarrow \mathcal{C}^\mathtt{L}$. Overall, $\mathtt{LQ}(\mathbf{h}) \triangleq \mathtt{LQD}( \mathtt{LQE}_1([\mathbf{f},\mathbf{g}_1]),\ldots, \mathtt{LQE}_L([\mathbf{f},\mathbf{g}_L]))$. Thus, $\mathtt{LQ}:\mathbb{C}^{R(K+L)}\rightarrow\mathcal{C}^\mathtt{L}$, and $\mathtt{LQ}(\mathbf{h})=\mathbf{x}$, for some $\mathbf{x}\in\mathcal{C}^\mathtt{L}$.
\subsection{Transmission Scheme}
We use a two-step AF protocol\cite{jing1, koyuncu1}. In the first step, the $k$th transmitter selects a symbol $s_k$ from a constellation $\mathcal{S}_k$, where $|\mathcal{S}_k|<\infty$, $\mathtt{P}(s_k) = |\mathcal{S}_k|^{-1},\,\forall s_k\in\mathcal{S}_k$, and sends $\sqrt{P_{S_k}}s_k$. We normalize $s_k$ as $\mathtt{E}[|s_k|^2]=1$. Thus, the average power used at the $k$th transmitter is $P_{S_k}$. During the first step, there is no reception at the receivers, but the $r$th relay receives
\begin{align}
t_r = \sum_{k=1}^K f_{kr}s_k\sqrt{P_{S_k}} + \eta_{0r},
\end{align}
where $\eta_{0r}\sim\mathtt{CN}(0,1)$.
%

Suppose that a quantizer $\mathtt{Q}:\mathbb{C}^{R(K+L)} \rightarrow \mathcal{C}$, global or local, is employed in the network, and $\mathtt{Q}(\mathbf{h}) = \mathbf{x}$, for some $\mathbf{x}\in\mathcal{C}$. Then, the relays use the beamforming vector $\mathbf{x}$ to adjust their transmit power and transmit phase. During the second step, the transmitters remain silent, but the $r$th relay transmits
\begin{align}
u_r = x_r \sqrt{\rho_r} t_r,
\end{align}
where the relay normalization factor $\rho_r$ is given by
\begin{align}
\rho_r \triangleq \frac{P_{R_r}}{1 + \sum_{i=1}^K |f_{ir}|^2P_{S_i} }.
\end{align}

The average power used at the $r$th relay can be calculated to be $\mathtt{E}_{s_1,\ldots,s_K,\eta_0r}[|u_r|^2] = |x_r|^2P_{R_r},\,\forall \mathbf{h}$. We require $0\leq|x_r|\leq 1$ as a result of the short term power constraint. The channel state dependent normalization factors $\rho_r$ ensure that the instantaneous transmit power of each relay remains within its power constraint with high probability.\footnote{Because of the noise at its received signal, a relay can exceed its transmit power constraint at some instants. The phrase ``short-term'' comes from the observation that, regardless of the channel states, the relay always obeys its power constraint when its transmit power is averaged over the transmitted symbols and the noise.}

Also, note that within the restriction of $0 \leq |x_r| \leq 1$, $\rho_r$ is the maximal normalization factor that we can use. In other words, if a factor $\rho_r''$ satisfies $\rho_r'' > \rho_r$ for some $\mathbf{h}$, then it violates the short term power constraint. Still, one can employ another factor $\rho_r'$ with $\rho_r' \leq \rho_r,\,\forall\mathbf{h}$ (e.g. $\rho_r' = P_{R_r}/(1+\sum_{k=1}^K (1+|f_{kr}|^4)P_{S_k})$). We shall discuss later in Section \ref{alowbound} whether or not such a different choice of the normalization factor can improve the network performance.

After the two steps of transmission that has been described above, the received signal at the $\ell$th receiver can be expressed as:
\begin{align}
  \label{receivedk} y_{\ell} =  \sum_{k=1}^K  \sum_{r=1}^R x_r\sqrt{\rho_r}f_{kr}g_{r\ell}\sqrt{P_{S_i}}s_i + \sum_{r=1}^R x_r g_{r\ell}\sqrt{\rho_r} \eta_{0r}  + \eta_{1\ell},
\end{align}
where $\eta_{1\ell}\sim\mathtt{CN}(0,1)$ is the noise at the $\ell$th receiver. We assume that the noises $\eta_{0r},\,r=1,\ldots,R$, and $\eta_{1\ell},\,\ell=1,\ldots,L$ are independent.
\subsection{Performance Measure}
\label{performancemeasure}
The $\ell$th receiver attempts to decode the symbols of the transmitters with indices given by an arbitrary but fixed set $\mathcal{D}_{\ell}\subset\{1,\ldots,K\},\,\mathcal{D}_{\ell} \neq\emptyset$. As an example, for a network with $K=3$ and $L=2$, let $\mathcal{D}_1 = \{1,2\}$ and $\mathcal{D}_2 = \{2,3\}$. Then, the first receiver is interested only in the symbols of the first and the second transmitters, while the second receiver is interested only in the symbols of the second and the third transmitters. In general, we assume that $\bigcup_{\ell} \mathcal{D}_{\ell} = \{1,\ldots,K\}$. This guarantees that at least one receiver is interested in the symbols of the $k$th transmitter. In particular, for $K=1$, we have $\mathcal{D}_{\ell} = \{1\},\,\forall\ell$.

Let us call the vector of transmitted symbols $\mathbf{s}_{\ell} = [s_k]_{k\in\mathcal{D}_{\ell}}$ as the super-symbol relevant to the $\ell$th receiver, and $\widetilde{\mathbf{s}}_{\ell}$ be its decoded version. We say that an error event occurs at a receiver if it incorrectly decodes its desired super-symbol. In this case, the optimal decoder at the $\ell$th receiver is an individual maximum likelihood (ML) decoder\footnote{In the literature, the phrase ``individual'' usually refers to the cases in which the a posteriori probability is maximized over a single transmitter alphabet. Note that, in our case, the maximization is over the product alphabet $\mathcal{S}_{\ell}$ that represents the set of all super-symbols that the $\ell$th receiver is interested in.} given by $\widetilde{\mathbf{s}}_{\ell} = \arg\max_{\mathbf{s}_{\ell}'\in\mathscr{S}_{\ell}} \mathtt{P}(\mathbf{s}_{\ell}' | y_{\ell},\mathbf{x},\mathbf{h})$, where $\mathscr{S}_{\ell} = \prod_{k\in\mathcal{D}_{\ell}} \mathcal{S}_k$ is the relevant super-symbol alphabet. For a fixed channel state $\mathbf{h}$, and beamforming vector $\mathbf{x}$, let $\mathtt{SER}_{\ell}^{\mathtt{IML}}(\mathbf{x}, \mathbf{h}) \triangleq \mathtt{P}(\widetilde{\mathbf{s}}_{\ell} \neq \mathbf{s}_{\ell})$ denote the conditional super-symbol error rate (SER) of the $\ell$th receiver with the individual ML decoder.

Let us now define a single quantity that represents the SER performance of all the receivers. We define the conditional \textit{network error rate} (conditional NER, or CNER), denoted by $\mathtt{CNER}(\mathbf{x},\mathbf{h})$, as the probability that at least one receiver incorrectly decodes its desired super-symbol.

Our performance measure, the NER, is the expected value of the CNER. Given a quantizer $\mathtt{Q}$ global or local, the NER can thus be expressed as
\begin{align}
\label{quantperrf}
\mathtt{NER}(\mathtt{Q}) \triangleq \mathtt{E}_{\mathbf{h}}[\mathtt{CNER}(\mathtt{Q}(\mathbf{h}),\mathbf{h})].
\end{align}
\subsection{Diversity Measure}
Let us also define a unique diversity measure for our network. Let $P_{R_r} = p_{R_r} P,\,r=1,\ldots,R$, $P_{S_k} = p_{S_k} P,\,k=1,\ldots,K$, where $p_{S_k},p_{R,r} <\infty$. In other words, we allow the power constraint of each transmitting terminal to grow linearly with $P$. Then, the \textit{first-order} diversity achieved by a quantizer $\mathtt{Q}$ is given by
\begin{align}
\label{firstorderdiversity}
d_1(\mathtt{Q}) \triangleq  \lim_{P\rightarrow\infty} -\frac{\log \mathtt{NER}(\mathtt{Q})}{\log P}.
\end{align}
One problem with this conventional definition of diversity is that it fails to characterize the asymptotic effect of possible sub-linear $P$-dependent terms (e.g. logarithmic terms) in the error rate expression. In order to properly handle such cases, we define the \textit{second-order} diversity as
\begin{align}
\label{secondorderdiversity}
d_2(\mathtt{Q}) \triangleq \lim_{P\rightarrow\infty} -\frac{\log \mathtt{NER}(\mathtt{Q}) + d_1(\mathtt{Q})\log P}{\log\log P}.
\end{align}
Note that the first-order diversity is always positive, while the second-order diversity can be negative.

Now, the \textit{diversity} (gain) achieved by a quantizer $\mathtt{Q}$ is given by $d(\mathtt{Q}) \triangleq (d_1(\mathtt{Q}), d_2(\mathtt{Q}))$.

With these definitions, the asymptotic performance with a quantizer $\mathtt{Q}$, as $P$ grows to infinity, can be expressed as
\begin{align}
\mathtt{NER}(\mathtt{Q}) \cong \mathtt{G}_\mathtt{A}(P) (\log P)^{-d_2(\mathtt{Q})} P^{-d_1(\mathtt{Q})},
\end{align}
where the factor $\mathtt{G}_\mathtt{A}(P)$ is the \textit{array gain}. It is sublogarithmic in the sense that $\lim_{P\rightarrow\infty} \frac{\mathtt{G}_\mathtt{A}(P)}{\log P} = 0$. Also, we use it only when we compare the performance of two quantizers that provide the same diversity gain.

Finally, for two diversity gains $d = (d_1, d_2)$, and $d' = (d_1', d_2')$, we say that $d$ is higher than $d'$ (or $d>d'$) if either $d_1 > d_1'$ or $d_1=d_1',\,d_2>d_2'$.
%
\subsection{Problem Statement}
Our goal is to design the quantizer $\mathtt{Q}$, given a limited feedback rate, such that the NER is minimized. We consider this problem for both GQs and LQs.

To achieve our goal, we first determine the maximal possible diversity with our network model. Then, we design structured fGQs that can achieve this diversity. Finally, we use our observations on fGQs to systematically design fLQs that achieve maximal first order diversity, and then, vLQs that achieve maximal diversity.

We would like to note that, as demonstrated in \cite{koyuncu1}, the numerical optimization of our quantizers is always possible by using algorithms such as the Generalized Lloyd Algorithm \cite{linde1, fleming1}. These algorithms can be used to improve the array gain performance, or in some particular cases, the second-order diversity performance of our structured codebook designs. We will not consider such optimizations in this paper since they are straightforward.
\section{Lower Bounds on Quantizer Performance}
\label{alowbound}
Before we attempt to design a high-performance low-rate quantizer, it is natural to determine the best possible performance we can expect with \textit{any} quantizer. In this section, we find lower bounds on the NER for both relay-interference and relay-broadcast networks that hold for any quantizer $\mathtt{Q}$, global or local.

Let $\mathcal{X} = \{\mathbf{x}\in\mathbb{C}^R:\|\mathbf{x}\|_{\infty} \leq 1\}$ represent the set of all beamforming vectors. Then, we have
\begin{theorem}
\label{theorem1}
Let $\mathtt{Q}:\mathbb{C}^{R(K+L)}\rightarrow\mathcal{C}$ with $\mathcal{C}\subset\mathcal{X}$. Then, there are constants $0<C_1,C_2< \infty$ that are independent of both $P$ and $\mathtt{Q}$, such that for all $\mathtt{Q}$, and for all $P$ sufficiently large,
\begin{align}
\label{bestpossibleperformance}
\begin{array}{ll}
\displaystyle\vphantom{\sum_x} \mathtt{NER}(\mathtt{Q})  \geq C_1\frac{1}{P^R}, & K = 1, \\
\displaystyle\vphantom{\sum_x}\mathtt{NER}(\mathtt{Q}) \geq C_2\frac{\log^R P}{P^{R}}, & K > 1.
\end{array}
\end{align}
Moreover, the bounds in (\ref{bestpossibleperformance}) hold for any relay normalization factor $\rho_r'$ that satisfies $\rho_r' \leq \rho_r,\,\forall \mathbf{h}$.
\end{theorem}
\begin{proof}
Please see Appendix \ref{proofoftheorem1}.
\end{proof}
In other words, for relay-broadcast networks, the maximal diversity gain is $(R,0)$. Indeed, for a network with $K=L=1$, it was shown in \cite{koyuncu1} that diversity $(R,0)$ is achievable.

On the other hand, for relay-interference networks, the maximal diversity gain is $(R,-R)$. Since $(R,0)>(R,-R)$, interference results in a second order diversity loss in our network model.

Theorem \ref{theorem1} also shows that a different relay normalization factor $\rho_r'$ \textit{cannot} improve the diversity upper bounds, provided that it satisfies the short-term power constraint, and a codebook $\mathcal{C}\subset\mathcal{X}$ is employed. Thus, for the rest of this paper, we will only consider $\rho_r$ as our relay normalization factor.

An immediate question that stems from Theorem \ref{theorem1} is whether there exists finite rate quantizers that can achieve maximal diversity. In the next section, we construct an fGQ that provides an affirmative answer.
\section{Maximal Diversity with an fGQ}
\label{secglobal}
In order to determine an fGQ that can achieve maximal diversity, let us first determine, for any $K$, the optimal GQ given a fixed codebook with finite cardinality.
\begin{proposition}
\label{prop1}
Given a fixed codebook $\mathcal{C}$ with $|\mathcal{C}|<\infty$, the optimal GQ is given by $\mathtt{GQ}_\mathcal{C}^{\star}(\mathbf{h}) \triangleq \arg\min_{\mathbf{x}\in\mathcal{C}} \mathtt{CNER}(\mathbf{x}, \mathbf{h})$.
\end{proposition}
\begin{proof}
Let $\mathtt{Q}':\mathbb{C}^{R(K+L)} \rightarrow \mathcal{C}$. We have
\begin{align}
\mathtt{CNER}(\mathtt{GQ}_{\mathcal{C}}^{\star}(\mathbf{h}), \mathbf{h}) \leq \mathtt{CNER}(\mathtt{Q}'(\mathbf{h}), \mathbf{h})\implies \mathtt{NER}(\mathtt{GQ}_{\mathcal{C}}^{\star}) \leq \mathtt{NER}(\mathtt{Q}'),\,\forall\mathtt{Q}'.
\end{align}
Thus, $\mathtt{GQ}_{\mathcal{C}}^{\star}$ performs at least as good as any quantizer with codebook $\mathcal{C}$.
\end{proof}

Therefore, given that we employ an optimal GQ encoder given by Proposition \ref{prop1}, the GQ codebook uniquely determines the system performance. But, there is one complication: If we ever want to implement the optimal GQ encoder, we should be able to evaluate $\mathtt{CNER}(\mathbf{x},\mathbf{h})$, for any given $\mathbf{x}$ and $\mathbf{h}$. Unfortunately, a closed form characterization of the CNER is very difficult, if not impossible. For that reason, we design a suboptimal quantizer that, instead of the actual CNER, uses an upper bound on the CNER. Fortunately, this suboptimal quantizer will be powerful enough to achieve maximal diversity for any $K$.
\subsection{An Upper Bound on the CNER}
For the $\ell$th receiver, instead of the individual ML decoder described in Section \ref{performancemeasure}, suppose that we employ a joint ML decoder
$\hat{\mathbf{s}}_{\ell} \triangleq \arg\max_{\mathbf{s}'\in\mathscr{S}} \mathtt{P}(\mathbf{s}'|y_{\ell},\mathbf{x},\mathbf{h})$, where $\mathscr{S} = \prod_k \mathcal{S}_k$. Recall that, for the individual ML decoder at the $\ell$th receiver, the a posteriori probability was maximized over  $\prod_{i\in\mathcal{D}_{\ell}} \mathcal{S}_i$. For the joint ML decoder, the maximization is over $\prod_k \mathcal{S}_k$ at all the receivers.

Let $\mathtt{SER}_{\ell}^\mathtt{JML}(\mathbf{x},\mathbf{h}) \triangleq \mathtt{P}(\hat{\mathbf{s}}_{\ell} \neq \mathbf{s})$ denote the error rate of the joint ML decoder. Then, we have $\mathtt{SER}_{\ell}^{\mathtt{IML}}(\mathbf{x},\mathbf{h})  \leq \mathtt{SER}_{\ell}^\mathtt{JML}(\mathbf{x},\mathbf{h}),\,\forall\ell$. Also,
from (\ref{receivedk}),\footnote{Note that, in order to be able to perform ML decoding, the receivers should know which beamforming vector is used by the relays. In other words, for each $\mathbf{h}$, the receivers should know $\mathtt{Q}(\mathbf{h})$. This explains why we need to have a quantizer decoder at each receiver as well as each relay.}
\begin{align}
\label{desterrbound}
\textstyle \mathtt{SER}_{\ell}^\mathtt{JML}(\mathbf{x}, \mathbf{h})  & \leq \frac{1}{|\mathscr{S}|} \sum_{\substack{\mathbf{s}, \hat{\mathbf{s}}\in\mathscr{S} \\ \mathbf{s} \neq \hat{\mathbf{s}} }} \mathrm{Q} \bigl(\sqrt{2\smash[t]{\gamma_{\ell,\mathbf{s},\hat{\mathbf{s}}}(\mathbf{x}, \mathbf{h}) }}\, \bigr),
\end{align}
where
\begin{align}
\label{gammakssh}
 \gamma_{\ell,\mathbf{s},\hat{\mathbf{s}}}(\mathbf{x}, \mathbf{h}) & = \frac{| \sum_{k=1}^K(s_k - \hat{s}_k) \sqrt{P_{S_k}}\sum_{r=1}^R f_{k r}\sqrt{\rho_r}g_{r\ell}x_r |^2}{4(1 + \sum_{r=1}^R \rho_r |g_{r\ell}|^2 |x_r|^2)}.
\end{align}

In (\ref{desterrbound}) and (\ref{gammakssh}), the decoded symbol vector for each receiver is obviously different, i.e.  $\hat{\mathbf{s}}_\ell$, though we have omitted the dependence on $\ell$ for brevity. Furthermore, from now on, we shall omit the condition $\mathbf{s}, \hat{\mathbf{s}}\in\mathscr{S}$ in the summations as it is clear from the context.

Now, using a union bound over all the receivers, it follows for the CNER that
\begin{align}
   \mathtt{CNER}(\mathbf{x},\mathbf{h}) &  \leq \sum_{\ell=1}^{L} \mathtt{SER}_{\ell}^{\mathtt{IML}}(\mathbf{x},\mathbf{h}) \\
  &  \leq \sum_{\ell=1}^{L} \mathtt{SER}_{\ell}^{\mathtt{JML}}(\mathbf{x},\mathbf{h}) \\
  &\label{zzxxss}  \leq  \frac{1}{|\mathscr{S}|}  \sum_{\ell=1}^{L}  \sum_{\mathbf{s}\neq\hat{\mathbf{s}}} \mathrm{Q} \bigl(\sqrt{2\smash[t]{\gamma_{\ell,\mathbf{s},\hat{\mathbf{s}}}(\mathbf{x}, \mathbf{h}) }}\, \bigr).
\end{align}
This upper bound can easily be evaluated for any constellation and thus, it is good enough for our purposes. However, for clarity of exposition in the rest of the paper, we seek a much simpler bound. First, let us define
\begin{align}
 \gamma_\ell(\mathbf{x}, \mathbf{h}) \triangleq
\min_{\mathbf{s} \neq \hat{\mathbf{s}} }
\gamma_{\ell,\mathbf{s},\hat{\mathbf{s}}}(\mathbf{x}, \mathbf{h}),
\end{align}
and
\begin{align}
 \gamma^{\mathtt{L}}(\mathbf{x}, \mathbf{h}) &  \triangleq \min_\ell\gamma_\ell(\mathbf{x}, \mathbf{h}) \\ &  = \min_\ell\min_{ \mathbf{s} \neq \hat{\mathbf{s}} }
\gamma_{\ell,\mathbf{s},\hat{\mathbf{s}}}
(\mathbf{x}, \mathbf{h}).
\end{align}
Then, (\ref{zzxxss}) can be further bounded as
\begin{align}
\label{combstep1}  \mathtt{CNER}(\mathbf{x},\mathbf{h}) & \vphantom{e^{\bigl()}}\leq \frac{|\mathscr{S}|-1}{2}  \sum_{\ell=1}^{L} \mathrm{Q}\bigl(\sqrt{2
\gamma_{\ell}(\mathbf{x}, \mathbf{h})}\bigr) \\
& \leq \frac{L(|\mathscr{S}|-1)}{2}  \max_{\ell} \mathrm{Q}\bigl(\sqrt{2
\gamma_{\ell}(\mathbf{x}, \mathbf{h})}\bigr) \\
\label{combstep3} &  = 2C_0 \mathrm{Q}\bigl(\sqrt{2
\gamma^{\mathtt{L}}(\mathbf{x}, \mathbf{h})}\bigr) \\
\label{combstep4} & \leq C_0 \exp(-\gamma^{\mathtt{L}}(\mathbf{x}, \mathbf{h})),
\end{align}
where $C_0 \triangleq L(|\mathscr{S}|-1)/4$. In the derivation above, (\ref{combstep1}) follows since there are $|\mathscr{S}|(|\mathscr{S}|-1)/2$ distinct terms with $\mathbf{s}\neq\hat{\mathbf{s}}$. For (\ref{combstep4}), we have used the fact that $\mathrm{Q}(x) \leq \frac{1}{2}\exp(-\frac{x^2}{2})$.

We would like to note the similarity of (\ref{combstep3}) and (\ref{combstep4}) to the conventional error rate expressions for single user wireless communication systems. Actually, the term $\gamma^{\mathtt{L}}(\mathbf{x},\mathbf{h})$ can be interpreted as a network signal-to-noise ratio (NSNR) measure that characterizes the overall performance of the network.
\subsection{Diversity Analysis of the Relay Selection Scheme}
%
%

For $K=L=1$, we have shown in \cite{koyuncu1} that a feedback scheme based on relay selection can achieve diversity $(R,0)$. Here, we generalize this result to any $L$.

For $K > 1$, due to both multiuser interference and its manifestation in Theorem \ref{theorem1}, it is not clear whether diversity $(R,-R)$ would be achievable. The main goal of this section is to show that it is indeed achievable with a GQ that maximizes the NSNR, and surprisingly, again using a simple relay selection codebook.

The relay selection codebook can be defined as $\mathcal{C}_{\mathtt{S}} = \{\mathbf{e}_r:r=1,\ldots,R\}$, where $e_{rq} = 1$ for $q=r$, and $e_{rq} = 0$ for $q\neq r$. Then, for any $K$ and $L$, We define our fGQ as
\begin{align}
\label{globqR}
\mathtt{GQ}_{\mathcal{C}_{\mathtt{S}}}(\mathbf{h}) = \arg\max_{\mathbf{e}_r\in\mathcal{C}_{\mathtt{S}}}\gamma^{\mathtt{L}}(\mathbf{e}_r,\mathbf{h}),
\end{align}
where, for any relay selection vector $\mathbf{e}^r$, we have from (\ref{gammakssh}) that
\begin{align}
\label{relayselectionfunc}
\gamma^{\mathtt{L}}(\mathbf{e}_r,\mathbf{h}) = \frac{1}{4}\min_{\mathbf{s} \neq \hat{\mathbf{s}} }\min_{\ell}\frac{\left| \sum_{k=1}^K(s_k - \hat{s}_k) \sqrt{P_{S_k}}f_{k r} \right|^2|g_{r\ell}|^2P_{R_r}}{1 + \sum_{k=1}^K |f_{k r}|^2 P_{S_k} + |g_{r\ell}|^2 P_{R_r} }.
\end{align}
Note that $\mathtt{GQ}_{\mathcal{C}_{\mathtt{S}}}$ chooses the relay selection vector that maximizes the NSNR.

 In the following theorem, we show that, for both relay-broadcast and relay-interference networks, $\mathtt{GQ}_{\mathcal{C}_{\mathtt{S}}}$ achieves maximal diversity by finding an upper bound on the NER:
\begin{theorem}
\label{theorem2}
There are constants $0<C_3,C_4<\infty$ that are independent of $P$ such that for all $P$ sufficiently large,
\begin{align}
\begin{array}{ll}
\displaystyle\vphantom{\sum_x} \mathtt{NER}(\mathtt{GQ}_{\mathcal{C}_{\mathtt{S}}}) \leq C_3 \frac{1}{P^{R}}, & K = 1, \\
\displaystyle\vphantom{\sum_x} \mathtt{NER}(\mathtt{GQ}_{\mathcal{C}_{\mathtt{S}}}) \leq C_4 \frac{\log^R P}{P^R},& K > 1.
\end{array}
\end{align}
\end{theorem}
\begin{proof}
 Please see Appendix \ref{proofoftheorem2}.
\end{proof}

In other words, the relay selection scheme with an fGQ achieves maximal diversity for any $K$. It is remarkable that full diversity is achieved regardless of the number of transmitters and receivers.

Note that our selection scheme requires $\lceil \log_2 R \rceil$ feedback bits. With $\lceil \log_2 R_0 \rceil$ feedback bits, where $R_0\in\{1,\ldots,R-1\}$, diversity orders $(R_0,0)$ and $(R_0,-R_0)$ are achievable for $K=1$, and $K>1$, respectively, simply by considering the selection scheme for any fixed $R_0$ of the relays and disregarding the others.

In practical networks, we may not have a GQ that knows the entire CSI of the network. In such situations, we would like to characterize the achievable performance using LQ encoders that know only a part of the CSI.

\section{Diversity with LQs}
\label{seclocaldddd}
In the previous section, we showed that a GQ using relay selection can achieve full diversity. Motivated by this result, we expect that a relay selection based LQ will achieve high diversity orders. In this section, we design two such LQs: An fLQ that achieves maximal first-order diversity, and a vLQ that achieves maximal diversity. Both quantizers will have similar structures. We construct them using the \textit{localization} method\cite{koyuncu2}, in which we synthesize an LQ out of an existing GQ. The synthesized LQ and the GQ share the same codebook. For our particular quantization scheme, we use the GQ $\mathtt{GQ}_{\mathcal{C}_\mathtt{S}}$ in (\ref{globqR}) as the basis of our LQs. Since $\mathtt{GQ}_{\mathcal{C}_\mathtt{S}}$ is based on relay selection, all of our LQs will be based on relay selection as well\footnote{In principle, the localization method itself is applicable to any GQ with any codebook; it is not limited to relay selection based GQs. However, for a general GQ, it is very difficult to analytically determine the performance of the synthesized LQ. Therefore, we focus only on the localization of relay selection based GQs.}.
\subsection{Localization}
Let $\mathtt{LQ}_{\xi,N}^{(\mathtt{f}|\mathtt{v})}$ denote a generic localization of $\mathtt{GQ}_{\mathcal{C}_\mathtt{S}}$. For the synthesized quantizer $\mathtt{LQ}_{\xi,N}^{(\mathtt{f}|\mathtt{v})}$, the superscript indicates whether it is fixed-length ($\mathtt{f}$) or variable-length($\mathtt{v}$); and $\xi,\,N$ are design parameters that we shall specify later on. For a particular channel state $\mathbf{h}$, the components of the synthesized quantizer operate as follows:
\subsubsection{LQ Encoders}
For notational convenience, $\omega_{r\ell} =  \gamma_{\ell}^{\mathtt{L}}(\mathbf{e}_r, \mathbf{h})$. The $\ell$th LQ encoder calculates $\omega_{r\ell},\,r=1,\ldots,R$. In other words, it calculates its own contribution to the NSNR for all possible relay selection vectors. Then, it quantizes each of the possible contributions using a scalar quantizer
\begin{align}
\mathcal{N}(x) = \left\{ \begin{array}{rl} n,& \exists n\in\{0,\ldots,N-2\}\mbox{ such that }x\in[n\xi, (n+1)\xi), \\
N,& \mbox{otherwise.}\end{array}\right.,\,x\in\mathbb{R}.
\end{align}
Its output message is the concatenation of $R$ sub-messages $\mathcal{N}(\omega_{r\ell}),\,r=1,\ldots,R$.
\subsubsection{An Illustration of the LQ Encoders}
\label{lqencillust}
Let us now illustrate the operation of the LQ encoders with a simple example with $R=3$, and $L=2$, as shown in Fig. \ref{lencex}. For some fixed channel variances, power constraints, and channel state $\mathbf{h}'$, suppose that $w_{11} = 1.7$, $\omega_{21} =  0.8$, $\omega_{31} =  1.2$, $\omega_{12} = 0.28$, $\omega_{22} = 0.67$, and $\omega_{32} = 2.3$. In the figure, each of these local NSNR values are represented by a disk ($\bullet$) on the real axis. Since we are using an LQ, $\omega_{r1},\,r = 1,2,3$ can be calculated only by the first receiver, and similarly, $\omega_{r2},\,r = 1,2,3$ can be calculated only by the second receiver. Note that the GQ has access to all the local SNRS and in this example, selects the relay with index $\arg\max_{r\in\{1,2,3\}} \min_{\ell} \omega_{r\ell} = 3$.
\begin{figure}[h]
\centering
\scalebox{0.85}{\epsfig{file = 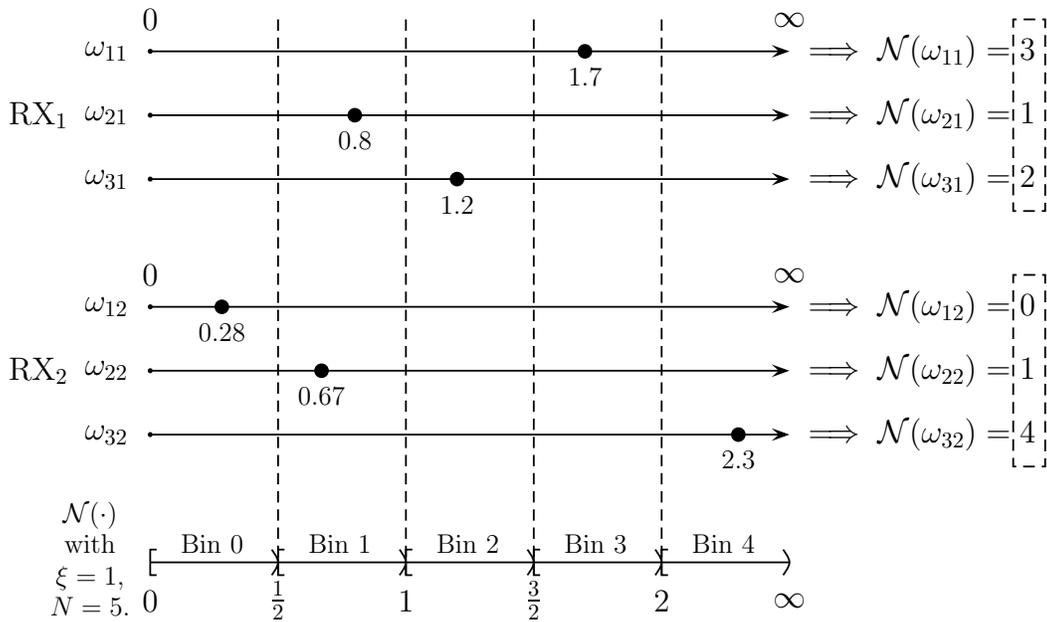, bbllx = 86pt, bblly = 414pt, bburx = 546pt, bbury = 686pt}}
\caption{An illustration of the LQ encoders.}
\label{lencex}
\end{figure}

After the LQ encoder calculates its local NSNR values, it quantizes them using a scalar quantizer $\mathcal{N}$ that is uniquely determined by the parameters $\xi$ and $N$. In our example, we use $N = 5$ bins and set $\xi = \frac{1}{2}$. Each bin is represented by a half open interval (\,\epsfig{file = 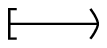, bbllx = 149pt, bblly = 389pt, bburx = 174pt, bbury = 395pt}\,) on the real axis. The output message of the LQ encoder is the concatenation of its quantized local NSNR values (submessages), shown as frames with a dashed outline, on the right hand side of the figure.
\subsubsection{Compressors}
\label{compressordef}
In general, there are $R$ sub-messages, each with $N$ possible values. Therefore, for a fixed-length synthesis $\mathtt{LQ}_{\xi,N}^{\mathtt{f}}$, at each channel state, each receiver feeds back $\lceil R\log_2 N \rceil$ bits without any compression.

For a variable-length synthesis $\mathtt{LQ}_{\xi,N}^{\mathtt{v}}$, we use a lossless compressor that produces an empty codeword (of length $0$) whenever $\mathcal{N}(\Omega_{r\ell})=N,\,\forall r$, and otherwise a codeword of length $\lceil \log_2 (N^R - 1) \rceil$ bits that can uniquely represent each $\mathcal{N}(\Omega_{r\ell})$. In other words, for a given channel state, the number of feedback bits produced by any receiver is either $0$ bits or $\lceil \log_2 (N^R - 1) \rceil$ bits\footnote{If the empty codeword is not allowed, one can use a ``$0$'' (a codeword of length $1$ bit) instead of the empty codeword, and append a ``$1$'' to each remaining codeword of length $\lceil \log_2 (N^R - 1) \rceil$ bits. The resulting codewords are uniquely decodable as well. Then, all of the results in this paper will hold for the case where the empty codeword is forbidden, given that the required feedback rates are increased by $1$ bit.

Also, note that one can achieve a better compression by using entropy encoders instead of the \textit{suboptimal} compressors that we employ. Even though the localization method was introduced originally with entropy encoders, the compressors that we use in this paper will be good enough for our purposes.}.

After all the $L$ feedback messages of the receivers are exchanged between the receivers and the relays, each of them decodes the feedback bits using the local decoder. The decoder operation is the same for each receiver and relay.
\subsubsection{Decompressor}
First, a decompressor perfectly recovers all the submessages from all the receivers, $\mathcal{N}(\omega_{r\ell}),\,r=1,\ldots,R,\,\ell = 1,\ldots,L$. All of these submessages are passed to the LQ decoder.
\subsubsection{An Illustration of the LQ Decoder}
For clarity of exposition, let us first present the LQ decoder for the example scenario in Section \ref{lqencillust} and the same channel state $\mathbf{h}'$. A more formal description of the general LQ decoder operation will be presented afterwards.

In general, the main goal of the LQ decoder is to imitate the GQ as good as possible. For our particular example, the GQ selects the relay with index $\arg\max_{r\in\{1,2,3\}} \omega_r$, where $\omega_r = \min\{\omega_{r1}, \omega_{r2}\}$. Then, the first goal of the LQ decoder should be to determine $\omega_r$. However, the LQ decoder only knows the quantized local NSNR values, $\mathcal{N}(\omega_{r\ell}),\,r=1,2,3,\,\ell = 1,2$, as shown in Fig. \ref{lencex}. Therefore, it cannot determine the exact value of $\omega_r$. However, as we shall describe in what follows, it can perfectly determine a subset of $\mathbb{R}$ where $\omega_r$ resides.

For any $\omega \in \mathbb{R}$, $\mathcal{N}(\omega) = n \implies \omega \in [\frac{n}{2}, \frac{n+1}{2}),\,n=0,\ldots,3$, and $\mathcal{N}(\omega) = 4 \implies \omega \in [2, \infty)$. We can use these facts to determine the possible locations of the local NSNR values, as represented in Fig. \ref{ldecex1} by half-open intervals (\,\epsfig{file = 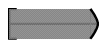, bbllx = 149pt, bblly = 389pt, bburx = 174pt, bbury = 395pt}\,)  of $\mathbb{R}$.
\begin{figure}[h]
\centering
\scalebox{0.85}{\epsfig{file = 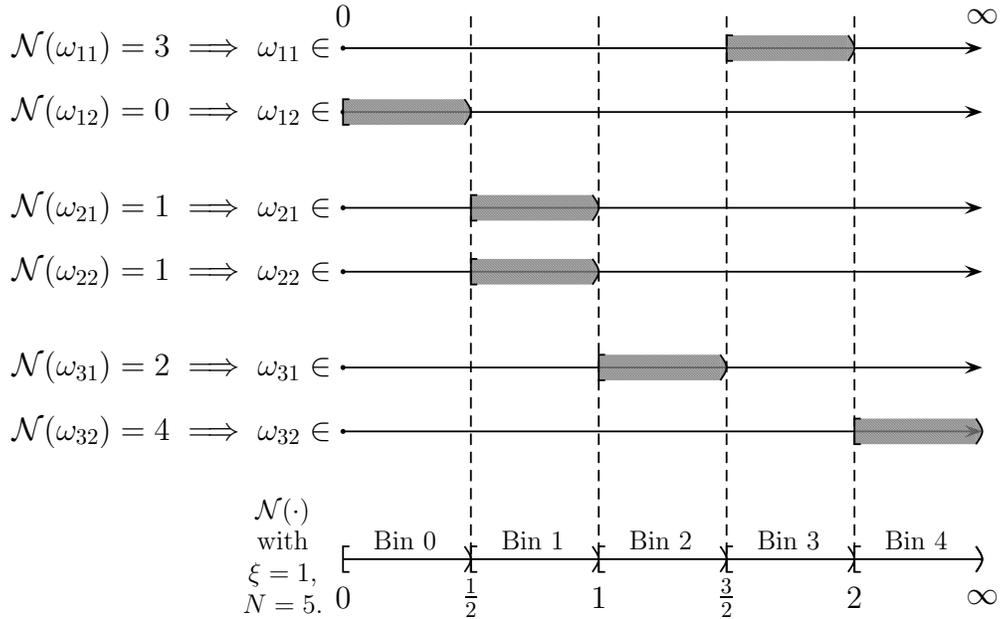, bbllx = 0pt, bblly = 414pt, bburx = 439pt, bbury = 686pt}}
\caption{Possible locations of the local NSNRs according to the LQ Decoder.}
\label{ldecex1}
\end{figure}

Since $\omega_1 = \min\{\omega_{11}, \omega_{12}\}$, and we know for sure that $\omega_{11}\in[\frac{3}{2}, 2)$ and $\omega_{12}\in[0, \frac{1}{2})$, we should have $\omega_1 \in [0, \frac{1}{2})$. Using the same arguments for all $r$, we can obtain $\omega_2 \in [\frac{1}{2},1)$, and $\omega_3 \in [1, \frac{3}{2})$. We have thus determined the possible locations of $\omega_r$, as shown in Fig. \ref{ldecex2}, by having access only to the quantized versions of $\omega_r$.
\begin{figure}[h]
\centering
\scalebox{0.85}{\epsfig{file = 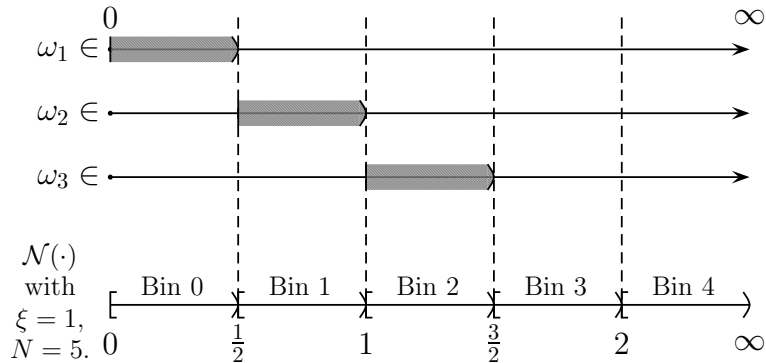, bbllx = 103pt, bblly = 527pt, bburx = 439pt, bbury = 686pt}}
\caption{Possible locations of the NSNRs according to the LQ Decoder.}
\label{ldecex2}
\end{figure}

The LQ decoder's main goal was to find $\arg\max_{r\in\{1,2,3\}} \omega_r$. Using the possible locations of $\omega_r$ that we have found, it is now clear that the third relay should provide the best NSNR. The LQ decoder's output will be $\mathbf{e}_3$. Note that this is the same output as the GQ output. Therefore, for this particular channel state, the LQ operates in the same manner as the GQ.

However, the LQ decoder will not be this lucky in general. As an example, another channel state might result in $\omega_1\in[0, \frac{1}{2})$ and $\omega_2, \omega_3\in[1, \frac{3}{2})$. In this case, the LQ decoder will know for sure that both the second relay and the third relay provides a larger NSNR than the first relay. On the other hand, it cannot determine which one of the second and the third relays provides the best NSNR. Therefore, it chooses one of them, and its decision may not be the optimal one that would instead be provided by the GQ. We shall quantify the effect of such suboptimal decisions later on.
%
%
\subsubsection{LQ Decoder}
 We now give the general and formal description of the LQ decoder.

 Let $\mathcal{R}_g \triangleq \{q:\min_{\ell} \omega_{q\ell} = \max_r\min_{\ell} \omega_{r\ell}\}$ denote the set of indices from which our GQ in (\ref{globqR}) produces its output.\footnote{$\mathcal{R}_g$ is not necessarily a singleton, but our definition of the $\arg\max$ guarantees that the GQ output is unique.} In other words, $\mathcal{R}_g$ is the set of indices of relays that provide the maximal NSNR. Also, let $\mathcal{R}_l \triangleq  \{q:\mathcal{N}(\min_{\ell} \omega_{q\ell}) = \mathcal{N}(\max_r\min_{\ell}\omega_{r\ell})\}$. Note that $\mathcal{R}_g \subset \mathcal{R}_l$. Moreover, due to the structure of $\mathcal{N}$, not only
\begin{align}
\mathcal{N}\bigl(\min_{\ell} \omega_{q\ell}\bigr)= \min_{\ell} \mathcal{N}(\omega_{q\ell}),
\end{align}
but also
\begin{align}
\mathcal{N}\bigl(\max_r\min_{\ell} \omega_{r\ell}\bigr) &  = \max_r \mathcal{N}\bigl(\min_{\ell} \omega_{r\ell}\bigr) \\ &  = \max_r \min_{\ell} \mathcal{N}(\omega_{r\ell}).
\end{align}
Therefore, $\mathcal{R}_l = \{q:\min_{\ell} \mathcal{N}(\omega_{q\ell}) = \max_r \min_{\ell} \mathcal{N}(\omega_{r\ell})\}$, and $\mathcal{R}_l$ can be easily calculated by the LQ decoder.

Since $\mathcal{R}_g \subset \mathcal{R}_l$, the LQ decoder can determine which relay selection vector(s) can possibly provide the maximal NSNR. In general, it can choose any one of the relay selection vectors that are indicated by $\mathcal{R}_{\ell}$. But, to be more precise, we define
\begin{align}
\textstyle \mathtt{LQ}_{\xi,N}^{(\mathtt{f}|\mathtt{v})}(\mathbf{h}) \triangleq \arg\max_{\mathbf{e}_r\in\mathcal{C}_\mathtt{S}} \min_{\ell} \mathcal{N}(\omega_{r\ell}).
\end{align}

\subsubsection{Localization Distortion}
Let us now study two possible cases of interest regarding the LQ output: If $\mathcal{R}_g = \mathcal{R}_l$, then the LQ output provides the same NSNR as the GQ output. Otherwise, the LQ might make a suboptimal decision. This results in what we call the \textit{localization distortion} (LD), given by
\begin{align}
\mathtt{LD}(\xi, N) \triangleq \mathtt{NER}(\mathtt{LQ}_{\xi,N}^{(\mathtt{f}|\mathtt{v})}) - \mathtt{NER}(\mathtt{GQ}_{\mathcal{C}_\mathtt{S}} ).
\end{align}

A useful upper bound on the LD can be calculated as:
\begin{align}
\mathtt{NER}(\mathtt{LQ}_{\xi,N}^{(\mathtt{f}|\mathtt{v})}) & \nonumber= \mathtt{E}_{\mathbf{h}}\bigl[\mathtt{CNER}(\mathtt{GQ}_{\mathcal{C}_\mathtt{S}}(\mathbf{h}), \mathbf{h}) \bigr| \mathcal{R}_{l} = \mathcal{R}_g\bigr] \mathtt{P}(\mathcal{R}_{l} = \mathcal{R}_g) + \\ &\qquad \qquad
\mathtt{E}_{\mathbf{h}}\bigl[\mathtt{CNER}(\mathtt{GQ}_{\mathcal{C}_\mathtt{S}}(\mathbf{h}), \mathbf{h}) \bigr| \mathcal{R}_{l} \neq \mathcal{R}_g\bigr] \mathtt{P}(\mathcal{R}_{l} \neq \mathcal{R}_g) \\
& \leq \mathtt{NER}(\mathtt{GQ}_{\mathcal{C}_\mathtt{S}} ) +
\mathtt{E}_{\mathbf{h}}\bigl[\mathtt{CNER}(\mathtt{GQ}_{\mathcal{C}_\mathtt{S}}(\mathbf{h}), \mathbf{h}) \bigr| \mathcal{R}_{l} \neq \mathcal{R}_g\bigr] \mathtt{P}(\mathcal{R}_{l} \neq \mathcal{R}_g) \\
& = \mathtt{NER}(\mathtt{GQ}_{\mathcal{C}_\mathtt{S}} ) +
\mathtt{E}_{\mathbf{h}}\bigl[\mathtt{CNER}(\mathtt{GQ}_{\mathcal{C}_\mathtt{S}}(\mathbf{h}), \mathbf{h}) \;\bigr|\;|\mathcal{R}_{l}| \geq |\mathcal{R}_g|\bigr] \mathtt{P}(|\mathcal{R}_{l}| \geq |\mathcal{R}_g|) \\
\label{handanhanim} & \leq \mathtt{NER}(\mathtt{GQ}_{\mathcal{C}_\mathtt{S}} ) +
\mathtt{LD}^{\mathtt{U}}(\xi,N),
\end{align}
where $\mathtt{LD}^{\mathtt{U}}(\xi,N)$ is the upper bound on the localization distortion, given by
\begin{align}
\label{uboundlocdist}
\mathtt{LD}^{\mathtt{U}}(\xi,N) \triangleq \mathtt{E}_{\mathbf{h}}\bigl[\mathtt{CNER}(\mathtt{GQ}_{\mathcal{C}_\mathtt{S}}(\mathbf{h}), \mathbf{h}) \bigr| |\mathcal{R}_{l}| \geq 2\bigr] \mathtt{P}(|\mathcal{R}_{l}| \geq 2).
\end{align}
\subsection{Maximal First-Order Diversity with an fLQ}
\label{seclocal}
Our main result concerning the fLQs is given by the following theorem:
\begin{theorem}
\label{theorem3}
Let $\xi_\mathtt{f} = \log^R P$, and $N_\mathtt{f} = 2$. Then, for $P$ sufficiently large, the NER with $\mathtt{LQ}_{\xi_\mathtt{f},N_\mathtt{f}}^{\mathtt{f}}$, which uses a fixed $R$ feedback bits per receiver per channel state, is upper bounded by
\begin{align}
\begin{array}{ll}
\displaystyle\vphantom{\sum_x} \mathtt{NER}(\mathtt{LQ}_{\xi_\mathtt{f},N_\mathtt{f}}^{\mathtt{f}}) \leq C_5 \frac{\log^{R}P}{P^R}, & K = 1, \\
\displaystyle\vphantom{\sum_x} \mathtt{NER}(\mathtt{LQ}_{\xi_\mathtt{f},N_\mathtt{f}}^{\mathtt{f}}) \leq C_6 \frac{\log^{2R}P}{P^R}, & K > 1.
\end{array}
\end{align}
where $0<C_5,C_6<\infty$ are constants that are independent of $P$.
\end{theorem}
\begin{proof}
Please see Appendix \ref{proofoftheorem3}.
\end{proof}
In other words, using a fixed $R$ feedback bits per receiver per channel state, we can achieve diversity $(R,-R)$ for $K=1$, and diversity $(R,-2R)$ for $K>1$. Since $(R,-R) < (R, 0)$ for the broadcast network, and $(R,-2R)<(R,-R)$ for the interference network, our fLQ has a second-order diversity loss compared to the optimal performance for both types of networks. Also, it is straightforward to show that, using $R_0$ bits, where $R_0 \in \{1,\ldots,R\}$, we can achieve diversity gains $(R_0, -R_0)$ and $(R_0, -2R_0)$ in relay-broadcast networks and relay-interference networks, respectively.

The scalar quantizer resolution for our fLQ is $\log_2 N_{\mathtt{f}} = 1$ bit per local NSNR. In what follows, we show that, by appropriately increasing the resolution with $P$, one can achieve maximal diversity, while the compressors make sure that the feedback rate remains bounded.
\subsection{Maximal Diversity with a vLQ}
For vLQs equipped with entropy coding, we have the following result:
\begin{theorem}
\label{theorem4}
Let $\epsilon > 0$ be a fixed constant that is independent of $P$. For any $\Lambda$ that satisfies $0 < \epsilon \leq \Lambda \leq P$, let $\xi_{\mathtt{v}} = \frac{1}{\Lambda}$, and
\begin{align}
\begin{array}{ll}
\displaystyle\vphantom{\sum_x}N_{\mathtt{v}} = \lceil \Lambda\log \Lambda + R \Lambda \log P +1 \rceil, & K = 1, \\
\displaystyle\vphantom{\sum_x}N_{\mathtt{v}} = \Bigl\lceil \Lambda\log \Lambda + R \Lambda \log\Bigl(\frac{P}{\log P}\Bigr)+1 \Bigr\rceil, & K > 1.
\end{array}
\end{align}
Then, for $P$ sufficiently large, we have
\begin{align}
\begin{array}{ll}
\displaystyle\vphantom{\sum_x} \mathtt{LD}^{\mathtt{U}}(\xi_{\mathtt{v}},N_{\mathtt{v}})  \leq C_7 \frac{1}{\Lambda P^R}, & K = 1, \\
 \displaystyle\vphantom{\sum_x} \mathtt{LD}^{\mathtt{U}}(\xi_{\mathtt{v}},N_{\mathtt{v}})  \leq C_8 \frac{\log^{R}P}{\Lambda P^R}, & K > 1,
 \end{array}
\end{align}
and, in addition, the feedback rate of the $\ell$th receiver satisfies
\begin{align}
\label{upboundrate}
\begin{array}{ll}
\displaystyle\vphantom{\sum_x}\mathtt{R}_{\ell}(\mathtt{LQ}_{\xi_\mathtt{v},N_\mathtt{v}}^{\mathtt{v}}) \leq C_{9} \frac{\log P}{P}, & K = 1,\\
\displaystyle\vphantom{\sum_x}\mathtt{R}_{\ell}(\mathtt{LQ}_{\xi_\mathtt{v},N_\mathtt{v}}^{\mathtt{v}}) \leq C_{10} \frac{\log^2 P}{P}, & K > 1,
\end{array}
\end{align}
where $0<C_7, C_8, C_9, C_{10}<\infty$ are constants that are independent of $\Lambda$ and $P$.
\end{theorem}
\begin{proof}
Please see Appendix \ref{proofoftheorem4}.
\end{proof}

We now describe several consequences of this theorem for $K>1$. The consequences for $K=1$ will be analogous.

Let us first recall from (\ref{handanhanim}) that $\mathtt{NER}(\mathtt{LQ}_{\xi,N}^{\mathtt{v}}) \leq \mathtt{NER}(\mathtt{GQ}_{\mathcal{C}_\mathtt{S}} ) +
\mathtt{LD}^{\mathtt{U}}(\xi_{\mathtt{v}},N_{\mathtt{v}})$. We have found an upper bound for $\mathtt{NER}(\mathtt{GQ}_{\mathcal{C}_\mathtt{S}} )$ in Theorem \ref{theorem2}. An upper bound for $\mathtt{LD}^{\mathtt{U}}(\xi_{\mathtt{v}},N_{\mathtt{v}})$ is given by Theorem \ref{theorem4}. Combining the two bounds, we have $\mathtt{NER}(\mathtt{LQ}_{\xi,N}^{\mathtt{v}}) \leq (C_4+C_8\Lambda^{-1})\frac{\log^R P}{P^R}$. In other words, our vLQ achieves maximal diversity.

Moreover, using the same arguments as in the previous paragraph, we have $\mathtt{NER}(\mathtt{LQ}_{\xi,N}^{\mathtt{v}}) \leq \mathtt{NER}(\mathtt{GQ}_{\mathcal{C}_\mathtt{S}}) +
\frac{C_4}{\Lambda}\frac{\log^R P}{P^R}$. Thus, by increasing $\Lambda$, the array gain performance of our vLQ can be made arbitrarily close to the one provided by the GQ, at any finite power level $P$.

What is more interesting is the behavior of the upper bound on the feedback rate given by (\ref{upboundrate}). As $P$ grows to infinity, the required feedback rate decays to zero. In other words, both the diversity and array gain benefits of $\mathtt{NER}(\mathtt{GQ}_{\mathcal{C}_\mathtt{S}})$ can be achieved using arbitrarily low feedback rates, when $P$ is sufficiently large.
\section{Simulation Results}
\label{secsims}
In this section, we present numerical evidence that verifies our analytical results. We assume that each receiver attempts to decode all the symbols from all the transmitters. In other words, $\mathcal{D}_{\ell} = \{1,\ldots,K\},\,\forall\ell$. In the graphs, ``GQ'' represents $\mathtt{GQ}_{\mathcal{C}_\mathtt{S}}$ in (\ref{globqR}), ``fLQ'' denotes $\mathtt{LQ}_{\xi_\mathtt{f},N_\mathtt{f}}^{\mathtt{f}}$ with $\xi_f$ and $N_f$ as defined in the statement of Theorem \ref{theorem3}. Also, ``vLQ-$\Lambda$'' represents $\mathtt{LQ}_{\xi_\mathtt{v},N_\mathtt{v}}^{\mathtt{v}}$ that is uniquely determined by the parameter $\Lambda$ as in the statement of Theorem \ref{theorem4}.

\subsection{Networks With Equal Parameters}
In Fig. \ref{fig1}, we show the performance results for a network with $K=R=L=2$, $\sigma_{f_{rk}}^2 = \sigma_{g_{r\ell}}^2 = p_{R_r} = p_{S_k} = 1,\,\forall r,k,\ell$, and $\mathcal{S}_1 = \mathcal{S}_2 = \{+1,-1\}$. For this network, the NERs with the GQ, fLQ, and vLQs for $\Lambda = 2^{-15},2^{-12},\ldots,2^{12}, 2^{15}$ is presented in Fig \ref{fig1a}. The horizontal and the vertical axes represent $P$ in decibels (dBs), and the NER, respectively.

We can observe that both our GQ and vLQs achieve the maximal diversity $(2,-2)$, while the fLQ achieves diversity $(2,-4)$. Moreover, as we increase $\Lambda$, the array gain performance of our vLQs can be made arbitrarily close to that of the GQ.
\begin{figure}
\centering
\subfloat[NERs.] {
\scalebox{1}{\epsfig{file = 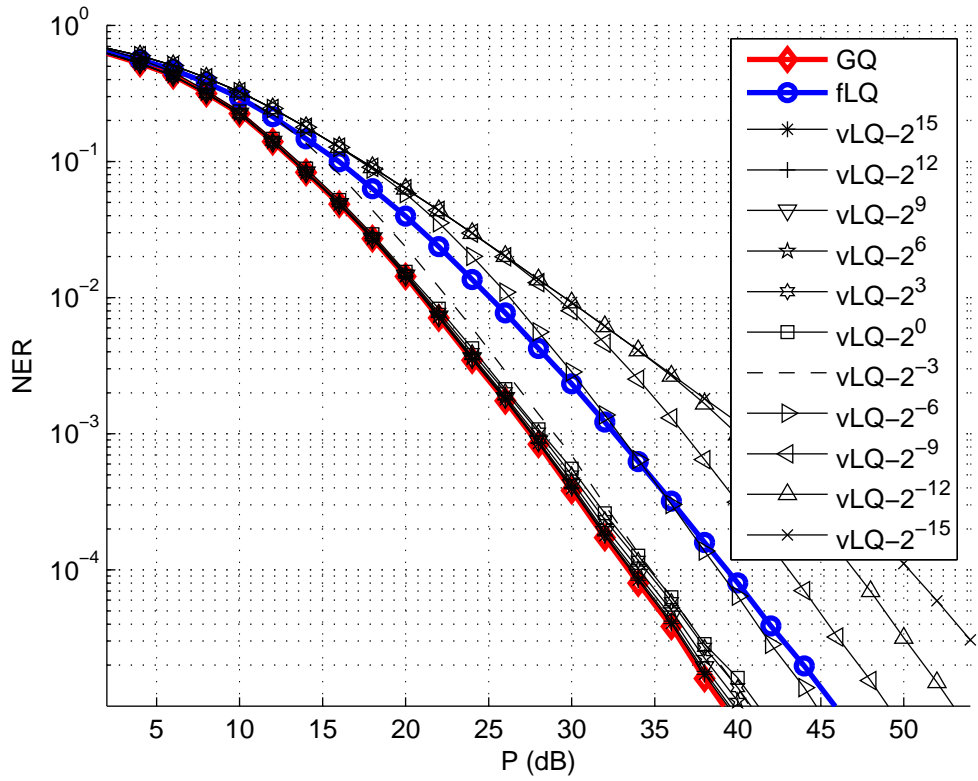, bbllx = 113pt, bblly = 245pt, bburx = 480pt, bbury = 560pt}}
\label{fig1a}
}
\vspace{-5pt}
\subfloat[SERs.] {
\scalebox{1}{\epsfig{file = 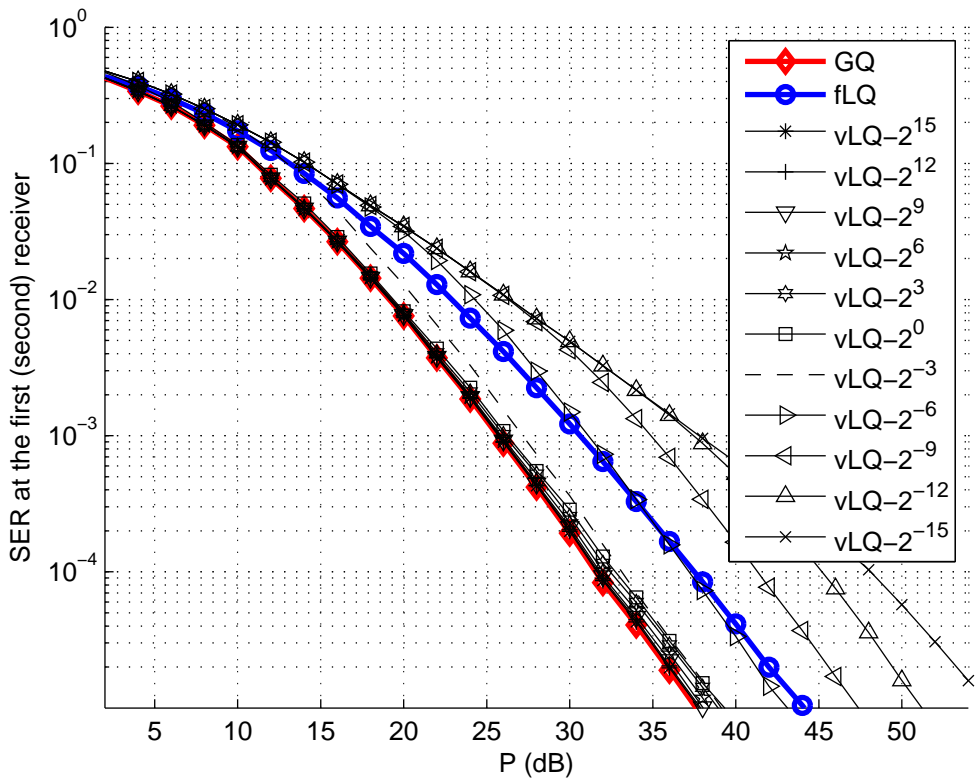, bbllx = 113pt, bblly = 245pt, bburx = 476pt, bbury = 560pt}}
\label{fig1b}
}
\caption{Performance results for a network with $K=R=L=2$.}
\end{figure}
\begin{figure}
\ContinuedFloat
\centering
\subfloat[Feedback rates.]{
\scalebox{1}{\epsfig{file = 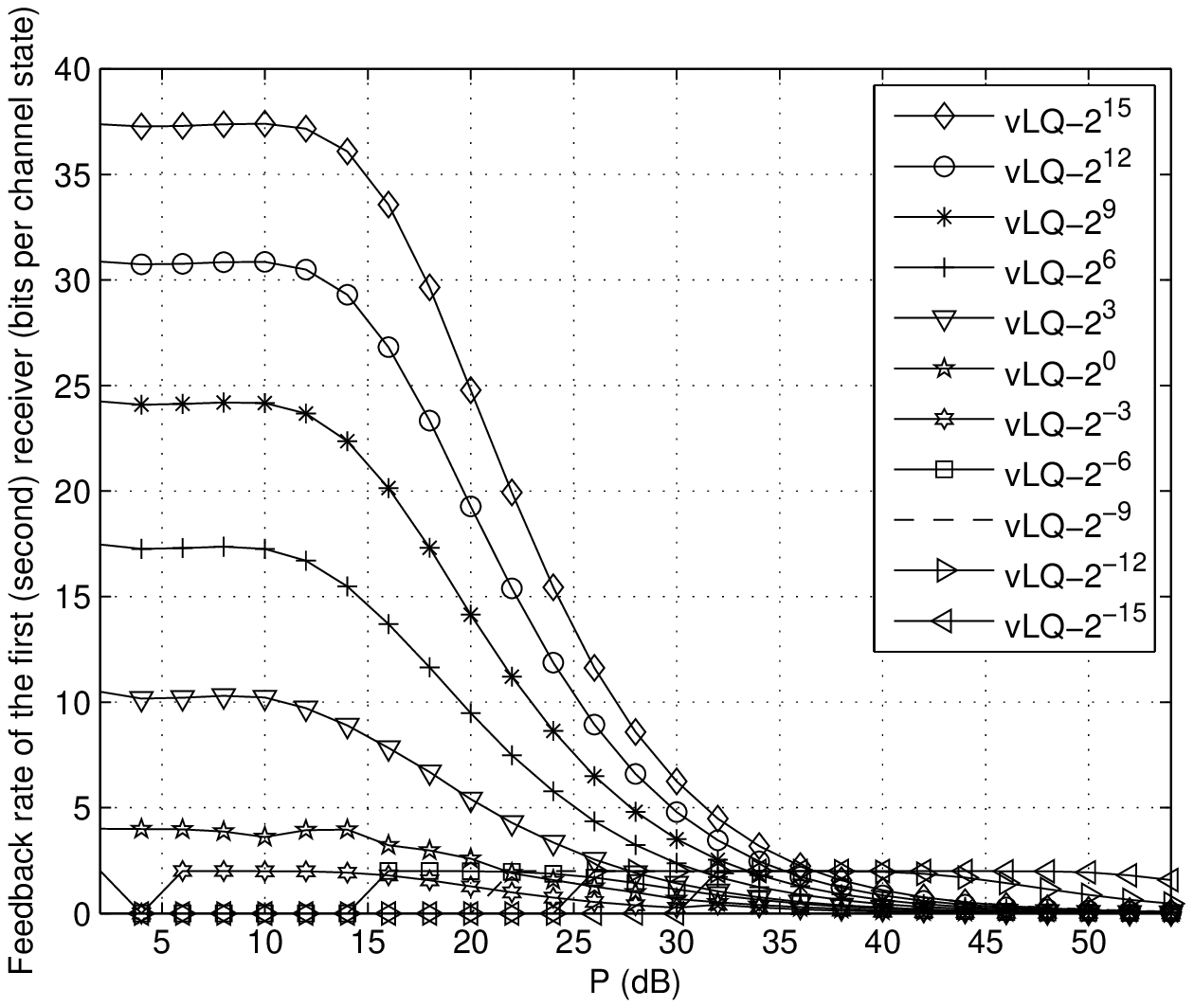, bbllx = 113pt, bblly = 245pt, bburx = 476pt, bbury = 560pt}}
\label{fig1c}
}
\caption{Performance results for a network with $K=R=L=2$ (continued).}
\label{fig1}
\end{figure}
%
%

In Fig. \ref{fig1b}, we show the SERs at the first receiver for the same network. The horizontal axis represents $P$ in decibels, while the vertical axis represents the SER at the first (second) receiver. As a result of our choice of network parameters, the SERs of each receiver is the same. Also, a particular quantizer achieves the same diversity as in Fig. \ref{fig1a}. On the other hand, since the SER is upper bounded by the NER, any quantizer in Fig. \ref{fig1b} provides more array gain than it does in Fig. \ref{fig1a}. Indeed, due to the symmetry of the network parameters, the SER performance is around $1.6$dB better than the NER performance for all quantizers.

The corresponding feedback rates of our vLQs are shown in Fig. \ref{fig1c}. The horizontal axis represents $P$ in decibels, while the vertical axis represents the feedback rate of the first (second) receiver in bits per channel state. Similarly, due to our choice of the network parameters, the feedback rates of each receiver will be the same. We can observe the validity of Theorem \ref{theorem4}, as for any $\Lambda$, the required feedback rate decays to zero at high $P$. Also, by increasing $\Lambda$, the performance of the LQs can be made arbitrarily close to the one provided by the GQ, while still using very low feedback rates. As an example, at an NER of $10^{-5}$, vLQ-$2^{15}$ needs $1.25$ bits per channel state per receiver on average and performs only $0.25$dB worse than the GQ. At a SER of $10^{-5}$, vLQ-$2^6$ uses $0.65$ bits, and GQ performs only $0.8$dB better.
%
\subsection{Networks With Unequal Parameters}
Our results also hold for networks with unequal power constraints and/or channel variances. To demonstrate that, we consider a network with $K=R=3$ and $L=4$. The parameters of the network are assumed to be $p_{S_1} = 1$, $p_{S_2} = 1.3$, $p_{S_3} = 0.7$, $p_{R_1} = 0.6$, $p_{R_2} = 2$, $p_{R_3} = 0.7$, $\mathcal{S}_1 = \mathcal{S}_3 = \{+1,-1\}$, and $\mathcal{S}_2 = \{e^{j\frac{\pi}{4}\theta}:\theta \in\{1,\ldots,4\}\}$. Also, we assume that $\sigma_{f_{kr}}^2 = F_{kr},\,\sigma_{g_{r\ell}}^2 = G_{r\ell},\,k=1,\ldots,K,\,r=1,\ldots,R,\,\ell = 1,\ldots,L$, where
\begin{align}
\mathbf{F} & = \left[\begin{array}{ccc} 2 & 1 & 0.7 \\ 1.5 & 0.9 & 3 \\ 1 & 4 & 0.5\end{array}\right],
 \end{align}
 and
 \begin{align}
\mathbf{G} & = \left[\begin{array}{cccc} 7 & 1.2 & 2.5 & 0.9 \\ 0.4 & 1.3 & 3 & 2 \\ 1.3 & 0.9 & 1.6 & 5 \end{array}\right].
\end{align}

In Fig. \ref{fig2a}, we show the NERs with the GQ, fLQ, and vLQs for $\Lambda = 2^{-15},2^{-12},\ldots,2^{0}, 2^{3}$. The results are analogous to what we have observed in Fig. \ref{fig1a}. Both the GQ and the vLQs achieve the maximal diversity $(3,-3)$, while the fLQ achieves diversity $(3,-6)$. Moreover, as we increase $\Lambda$, the array gain performance of our vLQs can be made arbitrarily close to that of the GQ.

The SERs at the first and the third receiver are shown in Fig. \ref{fig2b} and Fig. \ref{fig2c}. We can observe that, unlike the previous network with equal parameters, the SERs at each receiver is different for this network with unequal parameters. In particular, Fig. \ref{fig2b} reveals rather counterintuitive results: The fLQ outperforms the GQ at low $P$, and some of the vLQs provide a higher array gain than the GQ. The reason of these behaviors is that the GQ is optimized with respect to the NER, which takes into account the SERs of \textit{all} the receivers. Therefore, as far as the SER at a receiver is concerned, one cannot claim the optimality of the GQ. For the NER, the GQ outperforms all the other quantizers, as shown in Fig. \ref{fig2a}.

For the vLQs, the feedback rates of the first and the third receivers are shown in Fig. \ref{fig2d} and Fig. \ref{fig2e}. For both figures, the feedback rates decay to zero as $P$ grows to infinity, verifying Theorem \ref{theorem4}.
\begin{figure}
\centering
\subfloat[NERs.] {
\scalebox{0.6}{\epsfig{file = 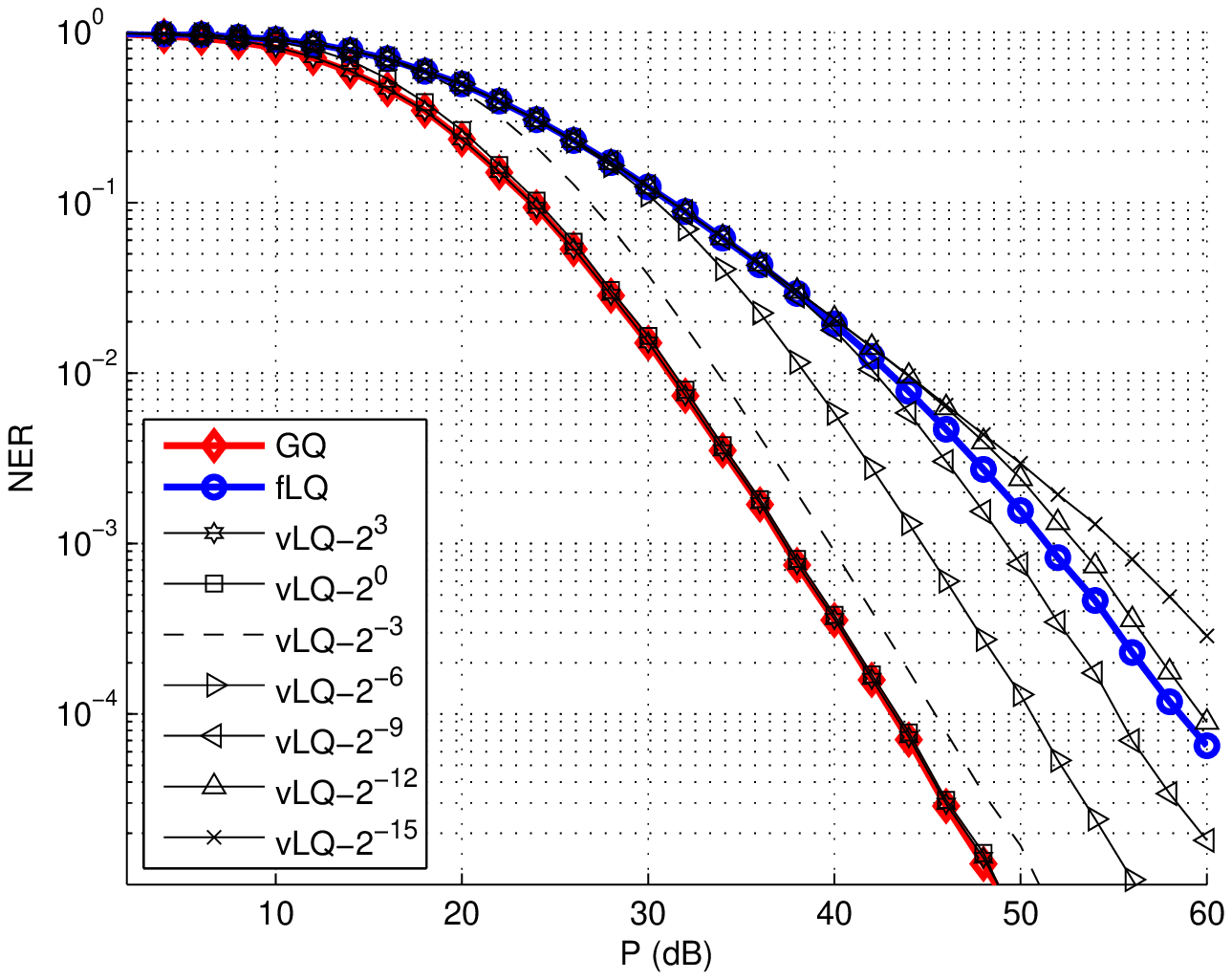, bbllx = 113pt, bblly = 245pt, bburx = 480pt, bbury = 560pt}}
\label{fig2a}
}

\subfloat[SERs at the first receiver.] {
\scalebox{0.6}{\epsfig{file = 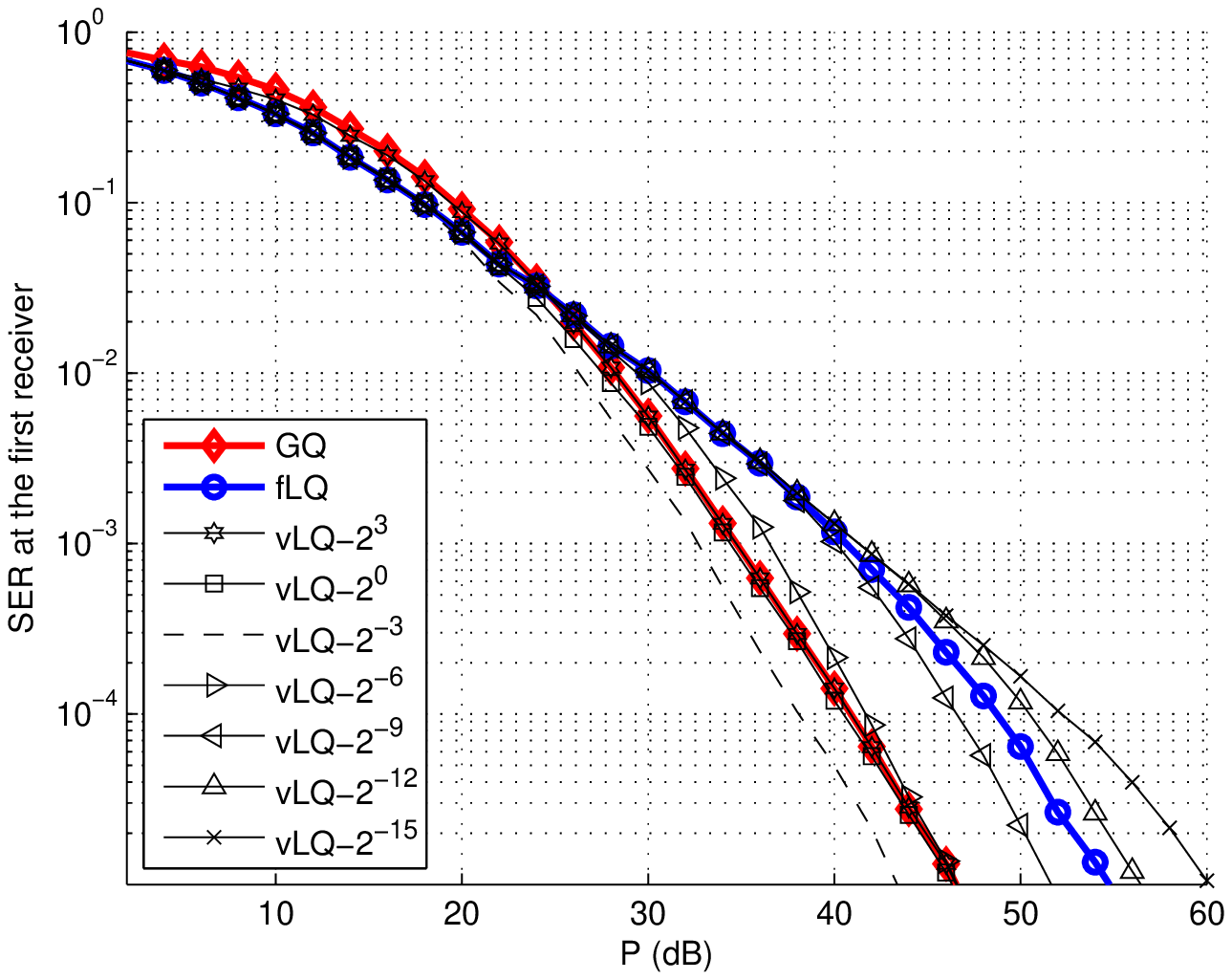, bbllx = 113pt, bblly = 245pt, bburx = 490pt, bbury = 560pt}}
\label{fig2b}
}
\subfloat[SERs at the third receiver.] {
\scalebox{0.6}{\epsfig{file = 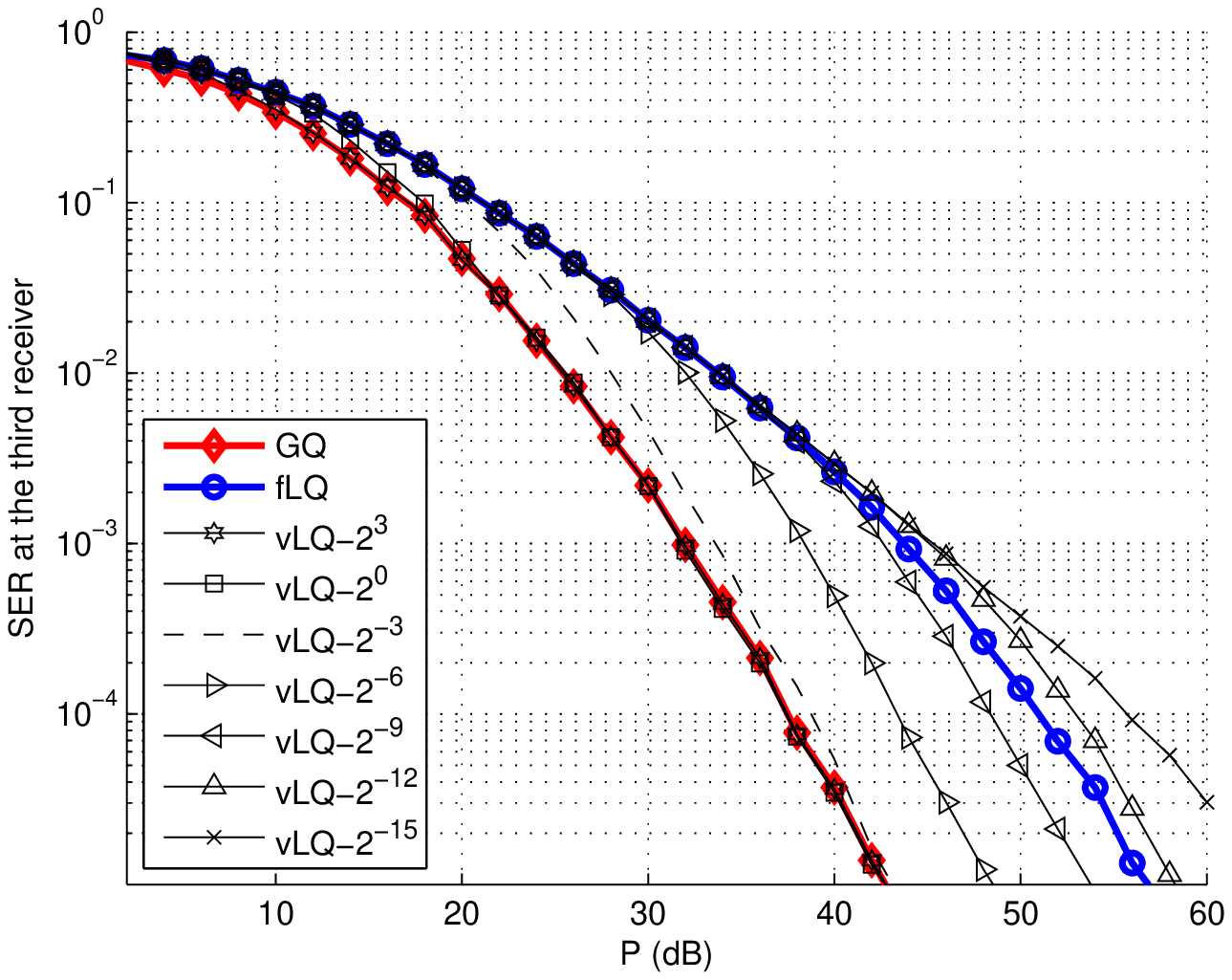, bbllx = 113pt, bblly = 245pt, bburx = 476pt, bbury = 560pt}}
\label{fig2c}
}

\subfloat[Feedback rates at the first receiver.] {
\scalebox{0.6}{\epsfig{file = 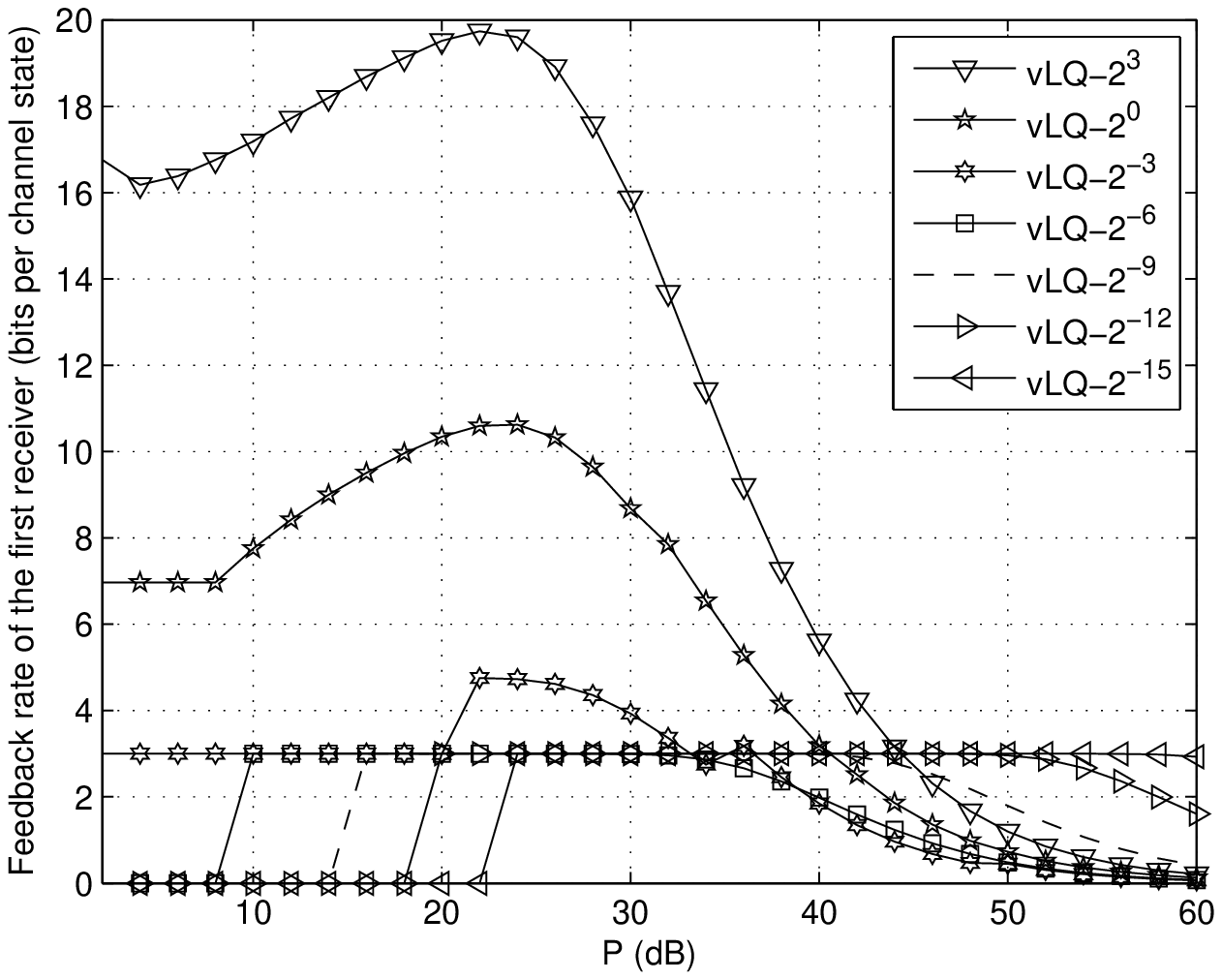, bbllx = 113pt, bblly = 245pt, bburx = 490pt, bbury = 560pt}}
\label{fig2d}
}
\subfloat[Feedback rates at the third receiver.] {
\scalebox{0.6}{\epsfig{file = 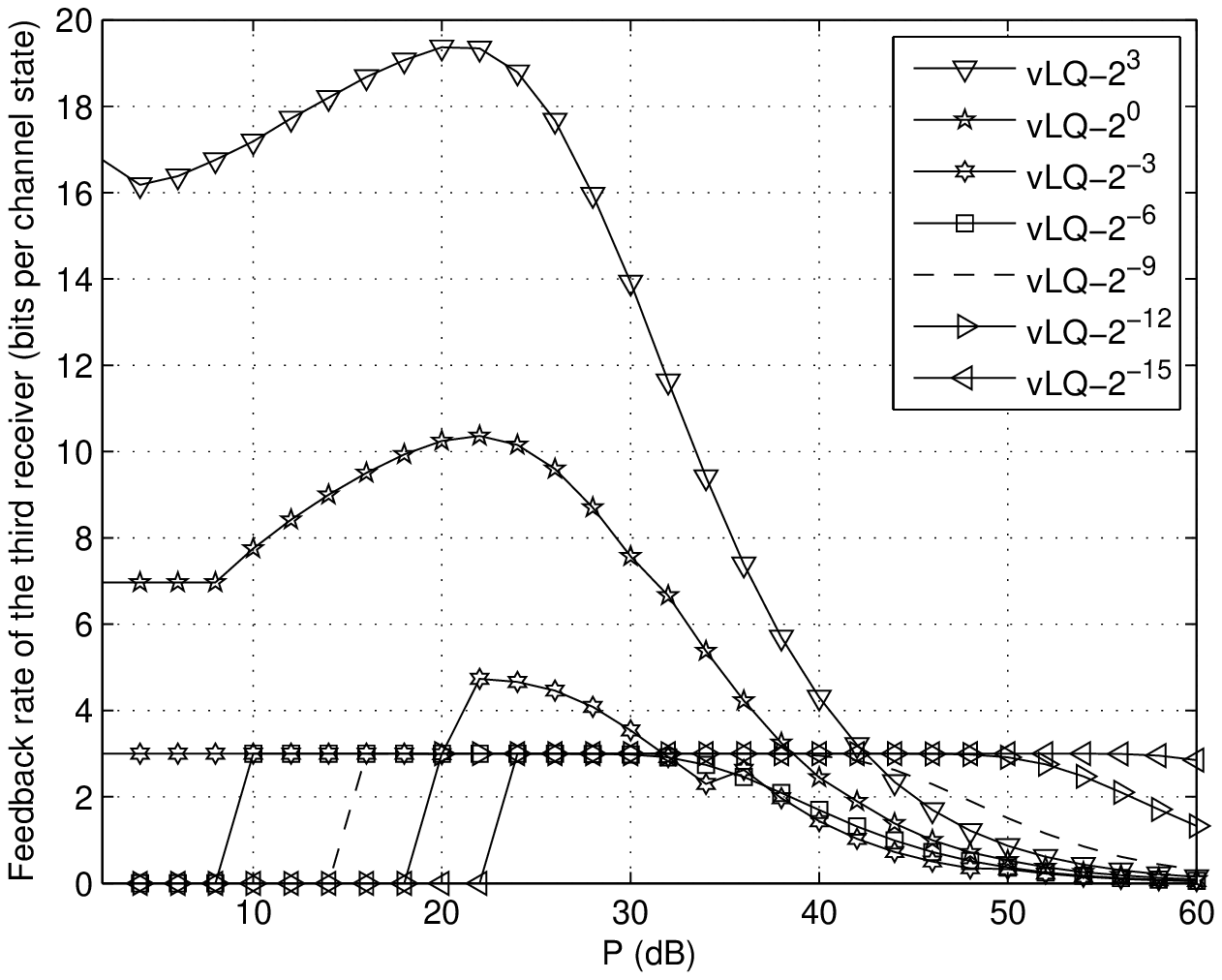, bbllx = 113pt, bblly = 245pt, bburx = 476pt, bbury = 560pt}}
\label{fig2e}
}
\caption{Performance results for a network with $K=R=3,\,L=4$.}
\label{fig2}
\end{figure}
\section{Conclusions and Discussions}
\label{secconcs}
We have studied quantized beamforming in wireless relay-interference networks with any number of transmitters, receivers and amplify-and-forward (AF) relays. Our goal has been to minimize the probability that at least one user incorrectly decodes its desired symbol(s).

We have introduced a generalized diversity measure in order to have a more precise description of the asymptotic performance of the network. It has encapsulated the conventional measure as the \textit{first-order} diversity. Additionally, it has taken into account the \textit{second-order} diversity, which is concerned with the transmitter power dependent logarithmic terms that appear in the error rate expression.

First, we have shown that, regardless of the quantizer and the amount of feedback that is used, interference results in a second-order diversity loss in our network model. Care should be taken though when making a general statement, as in this work, we have focused on AF networks with a short-term power constraint. For other forwarding methods, such as decode-and-forward, the diversity results may be different. Even under the restriction of using AF relays, one can use a long-term power constraint and achieve higher diversity. Also, the side information at the relays may be exploited for a better performance, though we believe this will not improve the diversity.

Second, we have designed a relay-selection based global quantizer (GQ) that can achieve maximal diversity. Then, using our GQ and the localization method, we have synthesized fixed-length and variable-length local quantizers (fLQs and vLQs). Our fLQ has achieved maximal first-order diversity. Our vLQ has provided not only maximal diversity gain, but also an array gain performance that can be made arbitrarily close to the one provided by the GQ. Moreover, it has achieved all of its promised gains using arbitrarily low feedback rates, when the transmitter powers are sufficiently large.

Regarding the LQs, there are many open problems that we have not addressed in this paper. One important problem is to determine whether there exists an fLQ that can achieve maximal diversity. Another goal might be to generalize our relay-selection based localization result to show that \textit{any} GQ can be localized to synthesize an LQ that can achieve the same array gain as the GQ. Due to the complicated nature of our distortion functions, the latter goal seems difficult to accomplish, even though we have observed its validity by simulations.
\appendices
\section{Upper Bounds on the PDF and CDF of $\Omega_r$}
\label{dendededen}
 First, let us present some useful lemmas.
\begin{lemma}
\label{lemmaunnamed}
Let $\widetilde{\mathbf{f}}_c$ and $\widetilde{\mathbf{f}}_s$ be zero-mean real Gaussian $K \times 1$ random vectors, with equal diagonal covariance matrices $\mathtt{E}[\widetilde{\mathbf{f}}_c \widetilde{\mathbf{f}}_c^T] = \mathtt{E}[\widetilde{\mathbf{f}}_s \widetilde{\mathbf{f}}_s^T] = \mathbf{K}$, $K_{ii}>0,\,\forall i$, $K_{ij} = 0,\,\forall i\neq j$, and zero cross-covariance $\mathtt{E}[\widetilde{\mathbf{f}}_c \widetilde{\mathbf{f}}_s^T] = \mathbf{0}$. Let $\widetilde{\mathbf{f}} \triangleq \widetilde{\mathbf{f}}_c + j\widetilde{\mathbf{f}}_s$ denote the complex Gaussian random vector with real and imaginary parts given by $\widetilde{\mathbf{f}}_c$ and $\widetilde{\mathbf{f}}_s$. Also, let $X = |\langle \mathbf{s},\widetilde{\mathbf{f}} \rangle |^2$, where $\mathbf{s}\in\mathbb{C}^K-\{\mathbf{0}\}$ is a fixed vector, and $W = \|\widetilde{\mathbf{f}}\|^2 - X$. Then, there is a constant $0 < \lambda_0 < \infty$, such that for all $x\geq 0$ and $w \geq 0$, we have
\begin{align}
 f_{X, W}(x,w) \leq \frac{\lambda_0^K}{\Gamma(K-1)} w^{K-2}\exp\bigr[-\lambda_0(x+w)\bigl].
\end{align}
\end{lemma}
\begin{proof}
 Let $\mathbf{f} \triangleq \mathbf{U}\widetilde{\mathbf{f}}$ for a unitary matrix $\mathbf{U}$ that satisfies $\mathbf{e}_1 = \mathbf{U}\mathbf{s}$. Also, let
 $X_i = |f_i|^2$, and $\mathbf{X} = [X_1 \cdots X_K]^T$. Note that
 \begin{align}
 X = X_1 = |\langle \mathbf{s}, \widetilde{\mathbf{f}}\rangle|^2 = |\langle \mathbf{e}_1, \mathbf{f}\rangle|^2 ,
 \end{align}
 and since $\mathbf{U}$ is norm-preserving,
 \begin{align}
 \textstyle W = \|\widetilde{\mathbf{f}}\|^2 - X = \|\mathbf{f}\|^2 - X_1 = \sum_{i=2}^K X_i. \end{align}
 Now, let $Y_1 = X_1$, $Y_2 = W = \sum_{i=2}^K X_i$, and $Y_k = X_k,\,k=3,\ldots,K$. Using such a transformation of RVs\cite{aofrvs}, we have
\begin{align}
\label{generalfxexp}
 f_{Y_1, Y_2}(y_1,y_2) = \int_0^{y_2} \int_0^{y_2-y_3} \cdots \int_0^{y_2-\sum_{k=3}^{K-1}y_k} f_{\mathbf{X}}(y_1, y_2-\textstyle \sum_{k=3}^Ky_k, y_3,\ldots,y_K)\mathrm{d}y_3\cdots \mathrm{d}y_K.
\end{align}
In the following, we find an upper bound for $f_{\mathbf{X}}(\mathbf{x})$ for any $\mathbf{x} = [x_1 \cdots x_K]^T$ with $x_i \geq 0,\,\forall i$.

Let $\mathbf{U}_c \triangleq \Re(\mathbf{U})$, and $\mathbf{U}_s \triangleq \Im(\mathbf{U})$. The real and imaginary parts of $\mathbf{f}$ can be calculated to be $\mathbf{f}_c \triangleq \Re(\mathbf{f}) = \mathbf{U}_c\widetilde{\mathbf{f}}_c - \mathbf{U}_s\widetilde{\mathbf{f}}_s$, and $\mathbf{f}_s \triangleq \Im(\mathbf{f}) = \mathbf{U}_c\widetilde{\mathbf{f}}_s + \mathbf{U}_s\widetilde{\mathbf{f}}_c$. Then, it is straightforward to show that
\begin{align}
\mathbf{K}_{cc} & \triangleq \mathtt{E}[\mathbf{f}_c\mathbf{f}_c^T] = \mathbf{U}_c \mathbf{K} \mathbf{U}_c^T + \mathbf{U}_s \mathbf{K} \mathbf{U}_s^T,  \\
\mathbf{K}_{ss} & \triangleq \mathtt{E}[\mathbf{f}_s\mathbf{f}_s^T] =
 \mathbf{K}_{cc},
\end{align}
and
\begin{align}
\mathbf{K}_{cs} \triangleq \mathtt{E}[\mathbf{f}_c\mathbf{f}_s^T] = \mathbf{U}_c \mathbf{K} \mathbf{U}_s^T - \mathbf{U}_s \mathbf{K} \mathbf{U}_c^T.
\end{align}
Therefore, $\mathbf{K}_{cc}$ and $\mathbf{K}_{ss}$ are symmetric matrices, and $\mathbf{K}_{cs} = -\mathbf{K}_{cs}^T$. The latter implies that for any $\mathbf{x}\in\mathbb{R}^K$, $\mathbf{x}^T\mathbf{K}_{cs}\mathbf{x} = 0$. Using these facts, we now show that $\mathbf{K}_{cc} + j\mathbf{K}_{cs}$ is positive definite.

For any $\mathbf{x}\in\mathbb{C}^K$, we have
\begin{align}
\mathbf{x}^H(\mathbf{K}_{cc} + j\mathbf{K}_{cs})\mathbf{x} & = (\mathbf{x}_c^T - j\mathbf{x}_s^T)(\mathbf{K}_{cc} + j\mathbf{K}_{cs})(\mathbf{x}_c + j\mathbf{x}_s) \\
\nonumber & = \mathbf{x}_c^T\mathbf{K}_{cc}\mathbf{x}_c - \mathbf{x}_c^T\mathbf{K}_{cs}\mathbf{x}_s +\mathbf{x}_s^T\mathbf{K}_{cs}\mathbf{x}_c + \mathbf{x}_s^T \mathbf{K}_{cc}\mathbf{x}_s + \\ & \qquad  j(\mathbf{x}_c^T\mathbf{K}_{cs}\mathbf{x}_c +\mathbf{x}_c^T\mathbf{K}_{cc}\mathbf{x}_s  - \mathbf{x}_s^T\mathbf{K}_{cc}\mathbf{x}_c + \mathbf{x}_s^T \mathbf{K}_{cs}\mathbf{x}_s) \\
& = \mathbf{x}_c^T\mathbf{K}_{cc}\mathbf{x}_c +2\mathbf{x}_s^T\mathbf{K}_{cs}\mathbf{x}_c + \mathbf{x}_s^T \mathbf{K}_{cc}\mathbf{x}_s \\
& = \mathbf{x}_c^T\mathtt{E}[\mathbf{f}_s\mathbf{f}_s^T]\mathbf{x}_c +2\mathbf{x}_s^T\mathtt{E}[\mathbf{f}_c\mathbf{f}_s^T]\mathbf{x}_c + \mathbf{x}_s^T \mathtt{E}[\mathbf{f}_c\mathbf{f}_c^T]\mathbf{x}_s \\
& = \mathtt{E}[(\mathbf{x}_c^T\mathbf{f}_s + \mathbf{x}_s^T\mathbf{f}_c)^2] \\
& = \mathtt{E}[((\mathbf{x}_c^T\mathbf{U}_s+ \mathbf{x}_s^T\mathbf{U}_c)\widetilde{\mathbf{f}}_c +
(\mathbf{x}_c^T\mathbf{U}_c - \mathbf{x}_s^T\mathbf{U}_s)\widetilde{\mathbf{f}}_s)^2] \\
& = \mathbf{x}_1^T \mathbf{K} \mathbf{x}_1 + \mathbf{x}_2^T \mathbf{K} \mathbf{x}_2,
\end{align}
where $\mathbf{x}_1 = \mathbf{U}_s^T\mathbf{x}_c+\mathbf{U}_c^T \mathbf{x}_s$, and $\mathbf{x}_2 = \mathbf{U}_c^T\mathbf{x}_c-\mathbf{U}_s^T \mathbf{x}_s$. But,
\begin{align}
\mathbf{U}\mathbf{x} = (\mathbf{U}_c^T + j\mathbf{U}_s^T)(\mathbf{x}_c + j\mathbf{x}_s) = \mathbf{x}_2 + j\mathbf{x}_1,
 \end{align}
and since $\mathbf{U}\mathbf{x}\neq 0$, either $\mathbf{x}_1\neq 0$ or $\mathbf{x}_2 \neq 0$. Also, since $\mathbf{K}$ is positive definite, either $\mathbf{x}_1^T \mathbf{K} \mathbf{x}_1 > 0$, or $\mathbf{x}_2^T \mathbf{K} \mathbf{x}_2 > 0$. Thus, $\mathbf{x}^H(\mathbf{K}_{cc} + j\mathbf{K}_{cs})\mathbf{x} > 0,\,\forall \mathbf{x}\in\mathbb{C}^K - \{\mathbf{0}\}$, and $\mathbf{K}_{cc} + j\mathbf{K}_{cs}$ is positive definite.

Let  $\mathbf{A} + j\mathbf{B} = (\mathbf{K}_{cc}+j\mathbf{K}_{cs})^{-1}$, and $\theta_{ij} = \tan^{-1}(B_{ij}/A_{ij})$. According to \cite[Eq. (24)]{mallik1}, the joint PDF of $\mathbf{X}_i,\,i=1,\ldots,K$ can be expressed as \begin{align}
 f_{\mathbf{X}}(\mathbf{x}) = (4\pi)^{-K}\mathrm{det}(\mathbf{A} + j\mathbf{B})\int_{-\pi}^{\pi}\cdots \int_{-\pi}^{\pi}\exp\Bigr[-\frac{1}{2}f(\boldsymbol{\phi}) \Bigl]\mathrm{d}\boldsymbol{\phi},
\end{align}
where
\begin{align}
\label{fksdhkfshdkfhskdhfksdf}
 f(\boldsymbol{\phi}) & =  \sum_{i=1}^K A_{ii}x_i + 2\sum_{\substack{i,j=1 \\ i<j}}^K(A_{ij}^2 + B_{ij}^2)^{\frac{1}{2}}\sqrt{x_ix_j}\cos(\phi_i - \phi_j + \theta_{ij}),
\end{align}
and $\mathbf{D} \triangleq \mathbf{A} + j\mathbf{B}$ is a Hermitian matrix\cite[Eq. (21)]{mallik1}.

 Since $f(\boldsymbol{\phi})$ is continuous, and the range of integration $[-\pi, \pi]^K$ is a compact subspace of $\mathbb{R}^K$, there exists $\boldsymbol{\phi^{\star}}\in\mathbb{R}^K$ with $\boldsymbol{\phi^{\star}} = [\phi_1^{\star} \cdots \phi_K^{\star}]$, such that $f(\boldsymbol{\phi}^{\star}) \leq f(\boldsymbol{\phi}),\,\forall \phi\in[-\pi, \pi]^K$. As a result,
 \begin{align}
 \label{ajdlskjdalksdlkasldkas}
f_{\mathbf{X}}(\mathbf{x}) \leq 2^{-K} \mathrm{det}(\mathbf{D})\exp\Bigr[-\frac{1}{2}f(\boldsymbol{\phi}^{\star}) \Bigl].
\end{align}
Now, let
\begin{align}
\mathbf{x}_{\mathtt{cos}} & = [\sqrt{x_1}\cos(\phi_1^{\star}) \cdots \sqrt{x_K}\cos(\phi_K^{\star})]^T, \\  \mathbf{x}_{\mathtt{sin}} & = [\sqrt{x_1}\sin(\phi_1^{\star}) \cdots \sqrt{x_K}\sin(\phi_K^{\star})]^T.
\end{align}
Then, using (\ref{fksdhkfshdkfhskdhfksdf}), $f(\boldsymbol{\phi}^{\star})$ can be expressed as
\begin{align}
\label{fboldymsoblksdlf}
f(\boldsymbol{\phi}^{\star}) = \Re(\mathbf{x}_{\mathtt{cos}}^T \mathbf{D} \mathbf{x}_{\mathtt{cos}} +  \mathbf{x}_{\mathtt{sin}}^T \mathbf{D} \mathbf{x}_{\mathtt{sin}} ).
\end{align}

We have shown that $K_{cc}+jK_{cs}$ is positive definite. It follows that $\mathbf{D} = (K_{cc}+jK_{cs})^{-1}$ is also positive definite, and thus has eigenvalues $\lambda_i > 0,\,\forall i$. Also, since $\mathbf{D}$ is a Hermitian matrix, it admits a spectral decomposition $\mathbf{D} = \sum_{i=1}^K \lambda_i \mathbf{u}_i \mathbf{u}_i^H$, where $\mathbf{u}_i,\,i=1,\ldots,K$ form an orthonormal basis for $\mathbb{C}^K$. It follows that
\begin{align}
 \mathbf{x}_{\mathtt{cos}}^T \mathbf{D} \mathbf{x}_{\mathtt{cos}} = \mathbf{x}_{\mathtt{cos}}^H \mathbf{D} \mathbf{x}_{\mathtt{cos}} & \textstyle = \sum_{i=1}^K \lambda_i (\mathbf{u}_i^H \mathbf{x}_{\mathtt{\cos}})^2 \\ \label{lamlamlam1} & > \lambda \|\mathbf{x}_{\mathtt{cos}}\|^2,
\end{align}
where $\lambda =  \min_i \lambda_i$. Similarly, we have
\begin{align}
\label{lamlamlam2} \mathbf{x}_{\mathtt{sin}}^T \mathbf{D} \mathbf{x}_{\mathtt{sin}} > \lambda \|\mathbf{x}_{\mathtt{sin}}\|^2.
\end{align}
Using (\ref{lamlamlam1}), (\ref{lamlamlam2}) and (\ref{fboldymsoblksdlf}), a lower bound on $f(\boldsymbol{\phi}^{\star})$ is given by
\begin{align}
f(\boldsymbol{\phi}^{\star})&  > \Re(\lambda \|\mathbf{x}_{\mathtt{cos}}\|^2 + \lambda \|\mathbf{x}_{\mathtt{sin}}\|^2)\\ & \label{fkdjfhskjdhfs}  = \lambda \textstyle \sum_{i=1}^K x_i,
\end{align}
Then, using (\ref{fkdjfhskjdhfs}) and (\ref{ajdlskjdalksdlkasldkas}), we can find an upper bound on  $f_{\mathbf{X}}(\mathbf{x})$ as
\begin{align}
 f_{\mathbf{X}}(\mathbf{x}) & \leq 2^{-K} \mathrm{det}(\mathbf{D})\textstyle\prod_{i=1}^K \exp(-\frac{\lambda}{2} x_i) \\ \label{kdsjhfksdjhfdskjhf} & \leq \lambda_0^K \textstyle\prod_{i=1}^K \exp(-\lambda_0 x_i),
\end{align}
where $\lambda_0 = \frac{\lambda}{2}$. For the last inequality, we have used the fact that $\det(\mathbf{D}) = \prod_{i=1}^K \lambda_i \leq \lambda^K$.

The lemma follows by substituting (\ref{kdsjhfksdjhfdskjhf}) to (\ref{generalfxexp}) and performing the integration.
\end{proof}
\begin{lemma}
\label{pdfcdfbound}
Let $X_1,\ldots,X_n$ be $n$ non-negative possibly dependent RVs, and $Z = \min_nX_n$. Then,
\begin{align}
\label{kolakutusu1}
f_Z(z) \leq \sum_{i=1}^n f_{X_i}(z),
\end{align}
and
\begin{align}
\label{kolakutusu2}
F_Z(z) \leq  \sum_{i=1}^n F_{X_i}(z).
\end{align}
\end{lemma}
\begin{proof}
Let us recall Leibniz's integral rule: For functions of a single variable $a(z)$, $b(z)$, and of two variables $f(x,z)$, we have
\begin{align}
 \frac{\partial}{\partial z}\int_{a(z)}^{b(z)}f(x,z)\mathrm{d}x=  \int_{a(z)}^{b(z)}\frac{\partial f}{\partial z}\mathrm{d}x+f(b(z),z)\frac{\partial b}{\partial z} -f(a(z),z)\frac{\partial a}{\partial z}.
\end{align}

Note that (\ref{kolakutusu2}) easily follows from (\ref{kolakutusu1}). We thus first prove (\ref{kolakutusu1}). Let $Z_k = \min\{X_1,\ldots,X_k\}$. We will show that $f_{Z_k}(z) \leq \sum_{i=1}^k f_{X_i}(z)$, for any $1\leq k\leq n$ by induction. For $k=1$, it is obvious. Suppose it is true for $n>k>1$. We have $f_{Z_k}(z) \leq \sum_{i=1}^k f_{X_i}(z)$. Noting that $Z_{k+1} = \min(Z_k, X_{k+1})$,
\begin{align}
  f_{Z_{k+1}}(z) & = f_{Z_k}(z) + f_{X_{k+1}}(z) - \frac{\partial}{\partial z} F_{Z_k,X_{k+1}}(z,z) \\
&  \leq \sum_{i=1}^{k+1} f_{X_i}(z) - \frac{\partial}{\partial z} \int_0^z \int_0^z f_{Z_k,X_{k+1}}(u,v) \mathrm{d}u \mathrm{d}v \\
   \label{leibuse1} &   = \sum_{i=1}^{k+1} f_{X_i}(z) - \int_0^z \left\{\frac{\partial}{\partial z} \int_0^z f_{Z_k,X_{k+1}}(u,v) \mathrm{d}v\right\} \mathrm{d}u - \int_0^z f_{Z_k,X_{k+1}}(z,v) \mathrm{d}v \\
 \label{leibuse2} &  =  \sum_{i=1}^{k+1} f_{X_i}(z) - \int_0^z f_{Z_k,X_{k+1}}(u,z) \mathrm{d}u - \int_0^z f_{Z_k,X_{k+1}}(z,v) \mathrm{d}v \\
 &   \leq \sum_{i=1}^{k+1} f_{X_i}(z),
\end{align}
where for both (\ref{leibuse1}) and (\ref{leibuse2}), we have used Leibniz's integral rule. This proves (\ref{kolakutusu1}). Integrating both sides of (\ref{kolakutusu1}) from $0$ to $z$ proves (\ref{kolakutusu2}).
\end{proof}
We can now find the desired upper bounds on the PDF and CDF of $\Omega_r$.
\begin{proposition}
\label{besselpdfcdfboundlemma}
For all $\omega>0,\,y\geq 0$ and $P$ sufficiently large,
\begin{enumerate}
\item If $K=1$,
\begin{align}
\label{upsilonrpdf}
f_{\Omega_r}(\omega) \leq \paa\exp\Bigl(-\pab\frac{\omega}{P}\Bigr)\psi_0(\omega),
\end{align}
and
\begin{align}
\label{upsilonrcdfdiff}
F_{\Omega_r}(y + \omega) - F_{\Omega_r}(y) \leq \paa \omega \psi_0(\omega),
\end{align}
where
\begin{align}
\psi_0(\omega) \triangleq \frac{1}{P}\left(1 + \omega^{-\frac{1}{2}}\right),
\end{align}
and $0< \paa, \pab < \infty$ are constants. Otherwise,
\item If $K > 1$,
\begin{align}
\label{besselpdfbound} f_{\Omega_r}(\omega) & \leq \lbb\exp\left(-\lbc \frac{z}{P}\right) \psi(\omega) , \\
\label{besseldiffcdfbound} F_{\Omega_r}(y + \omega)-F_{\Omega_r}(y) &  \leq \lbb\left(\omega\psi(\omega)+ \frac{\log P}{P^2}y^2\right),
 \end{align}
  and in particular, for $y=0$,
  \begin{align}
F_{\Omega_r}(\omega) \leq \lbb \omega\psi(\omega),
\end{align}
 where
 \begin{align}
\psi(\omega) = \frac{\log P}{P} \left(1 + \omega^{-\frac{1}{2\log P}} + \omega^{1-\frac{1}{\log P}}\frac{1}{P} \right),
 \end{align}
 and $0<\lbb,\lbc<\infty$ are constants.
\end{enumerate}
 \end{proposition}
\begin{proof}
First we prove the case for $K > 1$. Let $\Omega_{r,\ell,\mathbf{s},\hat{\mathbf{s}}} \triangleq \gamma_{\ell,\mathbf{s},\hat{\mathbf{s}}}^\mathtt{L}(\mathbf{e}_r, \mathbf{h})$. Note that $\Omega_r = \min_{\ell, \mathbf{s}\neq\mathbf{s}} \Omega_{r,\ell,,\mathbf{s},\hat{\mathbf{s}}}$. First, let us first find an upper bound on the PDF and CDF of $\Omega_{r,\ell,\mathbf{s},\hat{\mathbf{s}}}$.

Consider a fixed $r$, $\ell$, and $\mathbf{s}\neq\hat{\mathbf{s}}$. For notational convenience, let us define $Z \triangleq \Omega_{r,\ell,\mathbf{s},\hat{\mathbf{s}}}$. From (\ref{relayselectionfunc}), we have
\begin{align}
\label{alskdjalskdjlaksddddd}
Z = \frac{1}{4}\frac{\left| \sum_{k=1}^K(s_k - \hat{s}_k) \sqrt{P_{S_k}}f_{k r} \right|^2|g_{r\ell}|^2P_{R_r}}{1 + \sum_{k=1}^K |f_{k r}|^2 P_{S_k} + |g_{r\ell}|^2 P_{R_r} }.
\end{align}
Now, let us rewrite (\ref{alskdjalskdjlaksddddd}) in a more compact form. First, we define
\begin{align}
\mathbf{f}' & \triangleq [\sqrt{p_{S_1}}f_{1r} \cdots \sqrt{p_{S_1}}f_{Kr}]^T, \\
\boldsymbol{\delta} & \triangleq [s_1 - \hat{s}_1 \cdots s_K - \hat{s}_K]^H, \\
X & \triangleq |\langle \widetilde{\mathbf{f}'}, \boldsymbol{\delta} \rangle|^2, \\ Y & \triangleq |g_{r\ell}|^2p_{R_r},\\
W & \triangleq \|\mathbf{f}'\|^2 - X,
\end{align}
where $\widetilde{\mathbf{f}'} \triangleq \mathbf{f}'/\|\mathbf{f}'\|$. Then, we have
 \begin{align}
\label{dlaksdjlkas2} Z = \frac{\lbd XY P^2}{1 + XP+WP + YP },
\end{align}
where $\lbd = \frac{1}{4}\|\boldsymbol{\delta}\|^2$. Using a transformation of RVs\cite{aofrvs}, the PDF of $Z$ can be expressed as
\begin{align}
f_{Z}(z) & = \int_0^{\infty} \int_0^{\infty}f_{Y}\left(\frac{z(1+xP+wP)}{\lbd xP^2-zP}\right) \frac{\lbd xP^2(1+xP+wP)}{(\lbd xP^2-zP)^2}f_{X,W}(x,w)\mathrm{d}x\mathrm{d}w.
\end{align}
Now, let $\lbe = (\sigma_{g_{r\ell}}^2 p_{R_r})^{-1}$. Substituting the PDF of $Y$, and using Lemma \ref{lemmaunnamed}, we have
\begin{multline}
\label{dlaksdjlkas}f_{Z}(z) \leq \frac{\lbe \lbf^K}{\Gamma(K-1)}  \int_0^{\infty} \int_{\frac{z}{\lbd P}}^{\infty}\exp\Bigl(-\frac{\lbe z(1+xP+wP)}{\lbd xP^2-zP}\Bigr) \times \\ \frac{\lbd xP^2(1+xP+wP)}{(\lbd xP^2-zP)^2} e^{-\lbf x}\mathrm{d}x w^{K-2} e^{-\lbf w}\mathrm{d}w,
\end{multline}
where $\infty > \lbf > 0$ is a constant that is independent of $w$, $x$, and $P$. The inner integral can be evaluated first by a change of variables $u=\lbd xP-z$ and then using the facts that
$\int_0^{\infty} x^{\nu-1}e^{-\beta/x-\gamma x}\mathrm{d}x= 2(\beta/\gamma)^{\frac{\nu}{2}}K_{\nu}(2\sqrt{\beta\gamma})
,\beta,\gamma>0$\cite[Eq. 3.417.9]{toi}, and $K_{-1}(x) = K_1(x),\,\forall x\in\mathbb{R}$\cite[Eq. 9.6.6]{hmf}, respectively. Then, after some straightforward manipulations, we can rewrite (\ref{dlaksdjlkas}) as
\begin{align}
f_{Z}(z) & \leq \frac{\lbf^{K-1}\exp(-\lbg\frac{z}{P})}{\Gamma(K-1)}  \int_0^{\infty} \Bigl(\lbg P^{-1}\kappa K_1(\kappa) +\frac{2\lbf \lbe(\lbd+2z+\lbd wP)}{\lbd^2P^2}K_0(\kappa)  \Bigr) \frac{w^{K-2}}{e^{\lbf w}}\mathrm{d}w,
\end{align}
where $\kappa = \sqrt{4\lbf\lbe z(\lbd+z+\lbd wP)/(\lbd^2P^2)}$, and $\lbg = (\lbf+\lbe)/\lbd$. It follows that
\begin{align}
\label{kasjhdkjashd}
f_{Z}(z) & \leq \frac{\lbf^{K-1} \exp(-\lbg\frac{z}{P})}{\Gamma(K-1)} \int_0^{\infty} \Bigl(\lbg P^{-1} \kappa K_1(\kappa) +z^{-1}\kappa^2 K_0(\kappa)  \Bigr) w^{K-2}e^{-\lbf w}\mathrm{d}w.
\end{align}
Now, let us find an upper bound for $K_0(\kappa)$ in (\ref{kasjhdkjashd}). According to \cite[Eq. 9.6.24]{hmf}, we have $K_{\nu}(z) = \int_0^{\infty} e^{-z\cosh t} \cosh(\nu t)\mathrm{d}t,\,t,\nu\in\mathbb{R}$. Moreover, since $\cosh(\nu t)$ is an increasing function of $\nu$, $K_{\nu}(z)$ is also an increasing function of $\nu$. It follows that
\begin{align}
\label{k0icinbunlar1}
K_0(\kappa) \leq K_{\nu}(\kappa),\,\nu \geq 0. 
\end{align}
Also, from \cite[Eq. 25]{koyuncu1}, we have 
\begin{align}
\label{k0icinbunlar2}
K_\nu(\kappa) \leq 2^{\nu-1}\Gamma(\nu){\kappa}^{-\nu}, \, \nu>0.
\end{align}
Now let us set $0<\nu<1$. In this case,
\begin{align}
\label{k0icinbunlar3}
\Gamma(\nu) = \nu^{-1}\Gamma(\nu+1) = \nu^{-1}\smallint_0^\infty e^{-t^{1/\nu}}\mathrm{d}t \leq \nu^{-1}\smallint_0^\infty e^{-t}\mathrm{d}t = \nu^{-1}.
\end{align}
Combining (\ref{k0icinbunlar1}), (\ref{k0icinbunlar2}) and (\ref{k0icinbunlar3}) gives us the desired upper bound 
\begin{align}
\label{desiredboundfork0}
K_0(\kappa) \leq 2^{\nu-1}\nu^{-1}\kappa^{-\nu}. 
\end{align}
Using (\ref{desiredboundfork0}) and the fact that $\kappa K_1(\kappa) \leq 1$\cite[Eq. 25]{koyuncu1}, (\ref{kasjhdkjashd}) can be further bounded as
\begin{align}
f_{Z}(z) & \leq \frac{\lbf^{K-1} \exp(-\lbg \frac{z}{P})}{\Gamma(K-1)} \int_0^{\infty} \left(\lbg P^{-1} +2^{\nu-1}\nu^{-1}z^{-1}\kappa^{2-\nu}\right) w^{K-2}e^{-\lbf w}\mathrm{d}w \\
&  \leq \lbh \nu^{-1}e^{-\lbg\frac{z}{P}}\Biggl[P^{-1} + P^{-2+\nu} z^{-\frac{\nu}{2}}\int_0^\infty (1+wP+z)^{1-\frac{\nu}{2}}w^{K-2}e^{-\lbf w}\mathrm{d}w\Biggr],
\end{align}
where $\lbh = \lbf^{K-1}[\Gamma(K-1)]^{-1}\max\{\lbg, 2\lbf\lbe\max\{\frac{1}{\lbd},\frac{1}{\lbd^2}\},2(\lbf\lbe\max\{\frac{1}{\lbd},\frac{1}{\lbd^2}\})^{\frac{1}{2}}\}$.
Also, since $(x+y)^{\nu} \leq x^{\nu} + y^{\nu},\,\forall x,y\in\mathbb{R},\,0<\nu<1$, we have
\begin{align}
f_{Z}(z) & \leq  \lbh\nu^{-1}e^{-\lbg\frac{z}{P}}\Biggl[P^{-1} +  P^{-2+\nu} z^{-\frac{\nu}{2}}\int_0^\infty (1+w^{1-\frac{\nu}{2}}P^{1-\frac{\nu}{2}}+z^{1-\frac{\nu}{2}})w^{K-2}e^{-\lbf w}\mathrm{d}w\Biggr] \\
\nonumber &  = \frac{\lbh e^{-\lbg\frac{z}{P}}}{vP}\biggl[1 + z^{-\frac{\nu}{2}}P^{-1+\nu} (1+z^{1-\frac{\nu}{2}})\Gamma(K-1)\lbf^{-K+1} + \\ & \hspace{250pt} z^{-\frac{\nu}{2}}\Gamma\left(K - \frac{\nu}{2}\right)\lbf^{-K+\frac{\nu}{2}}P^{\frac{\nu}{2}}\biggr] \\
&  \leq 2\lbh\Gamma(K) \max\{1,\lbf^{-K+1},\lbf^{-K}\}e^{-\lbg\frac{z}{P}}\nu^{-1}P^{-1+\nu}\left(1 + z^{-\frac{\nu}{2}}+z^{1-\nu}P^{-1}\right) \\
\label{dangalakyus} &  = \lbi\frac{\log P}{P}e^{-\lbg\frac{z}{P}}\left(1 +z^{-\frac{1}{2\log P}}+z^{1-\frac{1}{\log P}}P^{-1}\right)
\end{align}
where $\lbi = 2e\lbh\Gamma(K) \max\{1,\lbf^{-K+1},\lbf^{-K}\}$, and we have substituted $\nu = \frac{1}{\log P}$ to obtain (\ref{dangalakyus}).

In general, the constants $\lbi$ and $\lbg $ in (\ref{dangalakyus}) depend on $r$, $\ell$, $\mathbf{s}$, and $\hat{\mathbf{s}}$. Let $\lbimod$ and $\lbg_{r,\ell,\mathbf{s},\hat{\mathbf{s}}}$ denote the dependent versions of $\lbi$ and $\lbg $, respectively. Using Lemma \ref{pdfcdfbound}, we have
\begin{align}
\label{zibidi}
 f_{\Omega_r}(\omega)  &  \leq \sum_{\ell}\sum_{\mathbf{s}\neq\mathbf{s}} \Omega_{r,\ell,\mathbf{s},\hat{\mathbf{s}}}(\omega) \\ &   \leq  \frac{\lbb}{4}\frac{\log P}{P}\exp\left(-\lbc\frac{z}{P}\right)\left(1 +z^{-\frac{1}{2\log P}}+z^{1-\frac{1}{\log P}}P^{-1}\right) ,
\end{align}
 where $\lbb = 2L\mathscr{S}|(|\mathscr{S}|-1)\max_{r,\ell,\mathbf{s}\neq\hat{\mathbf{s}}} \lbimod$, and $\lbc = \min_{r, \ell, \mathbf{s}\neq\mathbf{s}} \lbg_{r,\ell,\mathbf{s},\hat{\mathbf{s}}}$. This implies the upper bound on the PDF of $\Omega_r$ in the statement of the lemma. Finally, using (\ref{zibidi}),
\begin{align}
\label{asjdhilqwhyoieqoiw}   F_{\Omega_r}& (y+\omega) - F_{\Omega_r}(y)\\  &   = \int_y^{y+\omega} f_{\Omega_r}(\omega)\mathrm{d}\omega \\
&   \leq \frac{\lbb}{4}\frac{\log P}{P}  \int_y^{y+\omega}\left(1+x^{\frac{-1}{2\log P}} +  x^{1-\frac{1}{\log P}}\frac{1}{P} \right)\mathrm{d}x \\
&   \leq \frac{\lbb}{2}\frac{\log P}{P} \left(\omega+ (y+\omega)^{1-\frac{1}{2\log P}}-y^{1-\frac{1}{2\log P}} +   \left[(y+\omega)^{2-\frac{1}{\log P}} - y^{2-\frac{1}{\log P}}\right]\frac{1}{P} \right) \\
& \label{conconc}   \leq \frac{\lbb}{2}\frac{\log P}{P} \left(\omega + \omega^{1-\frac{1}{2\log P}} +  \left[y^{2-\frac{1}{\log P}} + 2\omega^{2-\frac{1}{\log P}}\right]\frac{1}{P}\right)\\
&   \leq \lbb\left(\omega\psi(\omega)+ \frac{\log P}{P^2}y^{2-\frac{1}{\log P}}\right),
\end{align}
where we have used H\"{o}lder's inequality and the fact that  $(y+\omega)^{\alpha} \leq y^{\alpha} + \omega^{\alpha},\,y,z>0,\,0\leq\alpha \leq 1$ for (\ref{conconc}). This concludes the proof for $K>1$.

For $K=1$, let $\bar{X}_r = |f_{1 r}|^2 p_{S_1}$, $\bar{Y}_r = p_{R_r} \min_{\ell} |g_{r\ell}|^2$, and $\Upsilon_r = \frac{ \bar{X}_r\bar{Y}_rP^2}{1 + \bar{X}_rP+ \bar{Y}_rP}$. From (\ref{relayselectionfunc}), we have $\Omega_r = \pac \Upsilon_r$, where $\pac = \frac{1}{4}\min_{s_1 \neq \hat{s}_1} (s_1 - \hat{s}_1)$.

Now, note that $\bar{X}_r \sim \Gamma(1, p_{S_1}\sigma_{f_{1r}}^2)$, and $\bar{Y}_r \sim \Gamma(1, p_{R_r}(\sum_{\ell} \sigma_{g_{r\ell}}^{-2})^{-1})$. Therefore, $\bar{X}_r,\,\bar{Y}_r,\,r=1,\ldots,R$ are independent exponential RVs with finite variances. The PDF of $\Upsilon_r$ with such $\bar{X}_r$ and $\bar{Y}_r$ is given by \cite[Eq. 22]{koyuncu1}. Using \cite[Eq. 28]{koyuncu1} without omitting the exponential function, and noting that $f_{\Omega_r} (\omega)= \frac{1}{\pac}f_{\Upsilon_r} (\frac{\omega}{\pac})$, we can show that (\ref{upsilonrpdf}) holds.

Finally, (\ref{upsilonrcdfdiff}) follows (up to a constant multiplier) from (\ref{upsilonrpdf}) and (\ref{asjdhilqwhyoieqoiw}). This concludes the proof.
\end{proof}
\section{Proof of Theorem \ref{theorem1}}
\label{proofoftheorem1}
We start with a lower bound on the CNER. By definition, we have $\mathtt{CNER}(\mathbf{x},\mathbf{h}) \geq \mathtt{SER}_{\ell}^{\mathtt{IML}}(\mathbf{x}, \mathbf{h}),\,\forall\ell$. Suppose that, for some $k\in\mathcal{D}_{\ell}$, a genie reveals all the transmitted symbols but $\mathbf{s}_k$ to the $\ell$th receiver. The error rate of this genie-aided scheme provides a lower bound on the CNER. Without loss of generality assume that $1\in\mathcal{D}_1$, and let us fix some $\acute{s}_1,\,\acute{s}_2\in\mathcal{S}_1$ with $\acute{s}_1\neq \acute{s}_2$. We have
$\mathtt{CNER}(\mathbf{x}, \mathbf{h}) \geq \frac{1}{|\mathcal{S}_1|} \mathrm{Q}(\sqrt{2\gamma^{\mathtt{U}}(\mathbf{x},\mathbf{h})})$, where
\begin{align}
 \gamma^{\mathtt{U}}(\mathbf{x},\mathbf{h}) \triangleq   \frac{|\sum_{r=1}^R f_{1 r}\sqrt{\rho_r'}g_{r1}x_r |^2|\acute{s}_1- \acute{s}_2|^2 P_{S_1}}{ 4(1 + \sum_{r=1}^R \rho_r' |g_{r1}|^2 |x_r|^2)}.
\end{align}
Let us find an upper bound on $\gamma^{\mathtt{U}}(\mathbf{x},\mathbf{h})$ for any $\mathbf{h}$, and $\mathbf{x}\in\mathcal{X}$. We have
\begin{align}
\gamma^{\mathtt{U}}(\mathbf{x}, \mathbf{h}) 
 & \leq \frac{| (\acute{s}_1- \acute{s}_2) \sqrt{P_{S_1}}\sum_{r=1}^R f_{1 r}\sqrt{\rho_r'}g_{r1}x_r |^2}{4(1 + \sum_{r=1}^R \rho_r' |g_{r1}|^2 |x_r|^2)}\\
& \leq \frac{|\acute{s}_1- \acute{s}_2|^2 P_{S_1} }{4}  \frac{|\sum_{r=1}^R f_{1 r}\sqrt{\rho_r'}\sqrt{R}g_{r1}\widetilde{x}_r|^2 }{R + \sum_{r=1}^R \rho_r' R |g_{r1}|^2 |\widetilde{x}_r|^2} \\
& \label{eq1111} =  \frac{|\acute{s}_1- \acute{s}_2|^2 P_{S_1} }{4}  \frac{\Bigl|\sum_{r=1}^R \frac{ f_{1 r}\sqrt{\rho_r'}\sqrt{R}g_{r1}}{\sqrt{1+ \rho_r' R |g_{r1}|^2}}\sqrt{1+ \rho_r' R |g_{r1}|^2}\widetilde{x}_r\Bigr|^2 }{ \sum_{r=1}^R (1+ \rho_r' R |g_{r1}|^2) |\widetilde{x}_r|^2},
 \end{align}
where $\widetilde{\mathbf{x}} \triangleq \frac{R}{\|\mathbf{x}\|}\mathbf{x}$ is the projection of the beamforming vector onto the hypersphere with norm $R$. Applying the Cauchy-Schwarz inequality to (\ref{eq1111}), and then using the fact that $\rho_r' \leq \rho_r$, we have
 \begin{align}
 \label{dalldalallalalaa}
  \gamma^{\mathtt{U}}(\mathbf{x}, \mathbf{h})
 & \leq \frac{|\acute{s}_1- \acute{s}_2|^2 P_{S_1} }{4}  \sum_{r=1}^R \frac{ |f_{1 r}|^2\rho_rR|g_{r1}|^2}{1+ \rho_r R |g_{r1}|^2}
 \end{align}
 If $K=1$, we use the following upper bound that follows from (\ref{dalldalallalalaa}).
  \begin{align}
  \gamma^{\mathtt{U}}(\mathbf{x}, \mathbf{h})
 & \leq \frac{|\acute{s}_1- \acute{s}_2|^2 P_{S_1} }{4}  \sum_{r=1}^R |f_{1 r}|^2.
 \end{align}
 This upper bound is, up to a constant multiplier, the same as the SNR of a maximal ratio combining system with $R$ branches. The error rate of such systems is known to be lower bounded by a constant times $P^{-R}$, as stated in the theorem. This concludes the proof for $K=1$.

 For $K>1$, we use (\ref{dalldalallalalaa}) to further bound $\gamma^{\mathtt{U}}(\mathbf{x}, \mathbf{h})$ as
 \begin{align}
 \gamma^{\mathtt{U}}(\mathbf{x}, \mathbf{h}) & \leq \frac{R^2 |\acute{s}_1- \acute{s}_2|^2 }{4}  \max_r \frac{ |f_{1 r}|^2|g_{r1}|^2P_{S_1}P_{R_r}}{1+ \sum_{k} |f_{kr}|^2 P_{S_k} +  R |g_{r1}|^2 P_{R_r}} \\
 & \leq \frac{R^2 |\acute{s}_1- \acute{s}_2|^2  }{4}  \max_r \frac{ |f_{1 r}|^2|g_{r1}|^2P_{S_1}P_{R_r}}{1+ |f_{1r}|^2 P_{S_1} +  |f_{2r}|^2 P_{S_2} + R |g_{r1}|^2 P_{R_r}} \\
 \label{okuzovichz} & \leq \frac{R^2 |\acute{s}_1- \acute{s}_2|^2\max_r\{\sigma_{f_{1r}}^2\sigma_{g_{r1}}^2p_{S_1} p_{R_r}\}  }{4\min\{1, \sigma_{f_{1r}}^2p_{S_1},\sigma_{f_{2r}}^2p_{S_2},\, \sigma_{g_{r1}}^2p_{R_r}\}} \max_r \frac{X_rY_rP^2}{1+ X_rP +W_rP + Y_rP},
 \end{align}
where $X_r = \sigma_{f_{1r}}^{-2} f_{1r}$, $Y_r = \sigma_{g_{r1}}^{-2} g_{r1}$, and $W_r = \sigma_{f_{2r}}^{-2} f_{2r}$. Note that $X_r,\,Y_r,\,W_r\sim\Gamma(1,1)$ and they are independent. Let $\taa$ denote the constant multiplier in (\ref{okuzovichz}), and $Z_r^{\mathtt{U}} \triangleq (X_rY_rP^2)/(1+ X_rP +W_rP + Y_rP)$. Thus, we can rewrite (\ref{okuzovichz}) as $\gamma^{\mathtt{U}}(\mathbf{x}, \mathbf{h}) \leq C_8 \max_r Z_r^{\mathtt{U}}$.
Now, let
\begin{align}
\textstyle Z^{\mathtt{U}} & \triangleq \textstyle\max_r Z_r^{\mathtt{U}}, \\
\textstyle \mathtt{NER}^{\mathtt{L}}(\mathtt{Q}) &\textstyle \triangleq \frac{1}{|\mathcal{S}_1|}\mathrm{E}[\mathrm{Q}(\sqrt{2\taa Z^{\mathtt{U}}})].
\end{align}

Since ${\mathtt{NER}}(\mathtt{Q}) \geq \mathtt{NER}^{\mathtt{L}}(\mathtt{Q}),\,\forall\mathtt{Q}$, it is sufficient to find a lower bound on $\mathtt{NER}^{\mathtt{L}}(\mathtt{Q})$. Using the fact that $\mathrm{Q}(x) \geq \frac{1}{\sqrt{2\pi}}\frac{x}{1+x^2}e^{-\frac{x^2}{2}}$, we have
 \begin{align}
\label{errlowboundxx}  \mathtt{NER}^{\mathtt{L}}(\mathtt{Q}) &  \geq \frac{1}{ |\mathcal{S}_1|\sqrt{\pi}} \int_0^{\infty} \frac{\sqrt{z}}{1+2z} \exp(-z\taa) f_{Z^{\mathtt{U}}}(z) \mathrm{d}z.
 \end{align}
We thus need to find a lower bound for the PDF of $Z^\mathtt{U}$. Using order statistics, we have
\begin{align}
\label{ordstatforsomeokuzrv}  f_{Z^{\mathtt{U}}}(z) =  \sum_{r=1}^R f_{Z_r^{\mathtt{U}}}(z) \prod_{\substack{q = 1\\ q\neq r}} ^R F_{Z_q^{\mathtt{U}}}(z).
\end{align}
 In the following, we find a lower bound on the PDF and CDF of $Z_r^\mathtt{U}$, for any $r$. We first evaluate the PDF of $Z_r^\mathtt{U}$. Using a transformation of RVs\cite{aofrvs}, it can be expressed as
  \begin{align}
f_{Z_r^{\mathtt{U}}}(z) & = \int_0^{\infty} \int_{0}^{\infty}f_{Y_r}\Bigl(\frac{z(1+xP+wP)}{ xP^2-zP}\Bigr) \frac{ xP^2(1+xP+wP)}{( xP^2-zP)^2} e^{-x}dx e^{-w}\mathrm{d}w.
\end{align}
This PDF is in the same form as (\ref{dlaksdjlkas2}) in Proposition \ref{besselpdfcdfboundlemma}, and can be evaluated using the same methods discussed therein. We have
  \begin{align}
 f_{Z_r^{\mathtt{U}}}(z) & = \exp\Bigr(-\frac{2z}{P}\Bigl) \int_0^{\infty} \biggl[ \frac{1}{P}\sqrt{\frac{4z(1+wP+z)}{P^2}} K_1\biggl(\sqrt{\frac{4z(1+wP+z)}{P^2}}\biggr) + \\ & \qquad \qquad \qquad \qquad  \frac{2(1+2z+ wP)}{P^2}K_0\biggl(\sqrt{\frac{4z(1+wP+z)}{P^2}}\biggr)  \biggr] e^{-w}\mathrm{d}w \\
& \geq \exp\Bigr(-\frac{2z}{P}\Bigl) \frac{1}{P^2} \int_0^\infty  2(1+wP+z) K_0\biggl(\sqrt{\frac{4z(1+wP+z)}{P^2}}\biggr) f_W(w)\mathrm{d}w.
\end{align}
Using the fact that for any $z>0$, $K_0(z) = -\log(\tfrac{z}{2})-\gamma_{e} + (1-\gamma_{e})\frac{1}{4z^2} + (1+1/2-\gamma_{e})\frac{z^4}{32} + \cdots$\cite{toi}, we have $K_0(z)  \geq -\log(\tfrac{z}{2})-\gamma_{e}$, and thus
\begin{align}
\textstyle f_{Z_r^{\mathtt{U}}}(z) &  \geq   \exp\Bigr(-\frac{2z}{P}\Bigl)\frac{1}{P^2} \int_0^\infty  (1+z+wP) \left[-\log\left(\frac{z(1+z+wP)}{P^2}\right) -2\gamma_e\right] e^{-w}\mathrm{d}w \\
&  = \nonumber \exp\Bigr(-\frac{2z}{P}\Bigl)  \frac{1}{P^2} \Biggl\{ (-2\gamma_e + 2\log P-\log z) \int_0^\infty (1+z+wP) e^{-w}\mathrm{d}w - \\ & \qquad \qquad \qquad \int_0^\infty (1+z+wP) \log(1+z+wP) e^{-w}\mathrm{d}w\Biggr\} \\
& = \nonumber \exp\Bigr(-\frac{2z}{P}\Bigl) \frac{1}{P^2} \biggl\{ (-2\gamma_e + 2\log P-\log z) (1+z+P) - \\ &  \qquad \Bigr[(1+z)\log(1+z) + P\log(1+z) + P + Pe^{\frac{1+z}{P}}\mathrm{E}_1\Bigr(\frac{1+z}{P}\Bigr)\Bigl] \biggr\}.
\end{align}
Using the facts that $\log z \leq \log (1+z) \leq z$, and
\begin{multline}
 e^{\frac{1+z}{P}}\mathrm{E}_1\left(\frac{1+z}{P}\right) \leq \log\left(1+\frac{P}{1+z}\right) \leq \log(1+z+P) \leq \log(1+z+P+zP) \\\textstyle  = \log(1+z)+\log(1+P) \leq z + \log(2P) \leq z+1+\log P,
\end{multline}
we can show that
\begin{align}
\textstyle f_{Z_r^{\mathtt{U}}}(z)  &  \geq  \exp\Bigr(-\frac{2z}{P}\Bigl)  \frac{1}{P^2} \Bigl\{ P\log P + 2\log P(1+z) - \\ & \qquad \qquad \bigl[(2\gamma_e + z)(1+z+P) +z(1+z)+2(z+1)P\bigr] \Bigr\} \\
&  \geq \exp\Bigr(-\frac{2z}{P}\Bigl)  \frac{1}{P^2} \Bigl\{ P\log P - P\bigl[(2 + z)(2+z)+ z(1+z) + 2(z+1)\bigr] \Bigr\} \\
\label{garagaragar} &  = \exp\Bigr(-\frac{2z}{P}\Bigl)  \frac{1}{P} \Bigl[\log P - (2z^2 + 7z + 6) \Bigr]
\end{align}
After some straightforward manipulations, (\ref{garagaragar}) leads to a more compact lower bound
\begin{align}
f_{Z_r^{\mathtt{U}}}(z)  \geq \phi(z),
 \end{align}
where
\begin{align}
 \phi(z) \triangleq e^{-\frac{2z}{P}}  \frac{1}{P} \bigl[ \log P - 14(1+z^2) \bigr].
\end{align}

For the CDF of $Z_r^\mathtt{U}$, we have
\begin{align}
 F_{Z_r^\mathtt{U}}(z) &   \geq \int_0^z \exp\Bigr(-\frac{2x}{P}\Bigl)  \frac{1}{P} \left[ \log P - 14(1+x^2)\right]\mathrm{d}x \\ &  \geq \exp\Bigr(-\frac{2z}{P}\Bigl)\frac{z}{P}\left[\log P - 14(1+z^2)\right] \\ & \textstyle  \geq z\phi(z).
\end{align}

We can now find a lower bound for the PDF of $Z^\mathtt{U}$. Suppose that $P \geq e^{14}$, and let $z_0 \triangleq(\tfrac{\log P}{14} - 1)^{\frac{1}{2}}$. Then, $\phi(z) \geq 0$ for $z \leq z_0$, and $\phi(z) < 0$, otherwise.
Using (\ref{ordstatforsomeokuzrv}), for $z \leq z_0$, it follows that
\begin{align}
\textstyle f_{Z^\mathtt{U}}(z) &  \geq Rz^{R-1}\phi^R(z) \\
& = Rz^{R-1} \exp\Bigl(-\frac{2Rz}{P}\Bigr)  \frac{1}{P^R} \left[\log P - 14(1+z^2)\right]^R\\
& \geq Rz^{R-1} \exp(-2Rz)  \frac{1}{P^R} \left[\log P - 14(1+z^2)\right]^R\\
& = Rz^{R-1} \exp(-2Rz)  \frac{1}{P^R} \sum_{i=0}^{R} {R \choose i} \log^{R-i}P \;(-14)^i\;(1+z^2)^i \\
 &  \geq Rz^{R-1} \exp(-2Rz)   \frac{1}{P^R} \Bigr[\log^R P-\sum_{i=1}^{R} {R \choose i} \log^{R-i}P\; 14^i\;(1+z^2)^i \Bigl] \\
\label{zlowbond}  &  \geq Rz^{R-1} \exp(-2Rz)  \frac{1}{P^R} \bigr[\log^R P-R2^{2R-1}14^R \log^{R-1}P (1+z^{2R}) \bigl]
\end{align}
 Since $f_{Z^\mathtt{U}}(z)$ is a PDF, $f_{Z^\mathtt{U}}(z) \geq 0,\,\forall z$. Therefore, for $z > z_0$, we can choose any negative function as a lower bound on $f_{Z^\mathtt{U}}(z)$. But, (\ref{zlowbond}) is negative for $z > z_0$. Thus, it is a lower bound on $f_{Z^\mathtt{U}}(z)$ that holds for all $z$. We can therefore use it to bound (\ref{errlowboundxx}) as
\begin{multline}
 \label{aljsdhlasd} \mathtt{NER}^\mathtt{L}(\mathtt{Q}) \geq \frac{R}{ |\mathcal{S}_1|\sqrt{\pi}}\frac{\log^R P}{P^R} \int_0^{\infty} \frac{z^{R-\frac{1}{2}}}{1+2z}  e^{-z(\taa +2R)} \mathrm{d}z - \\
\frac{R^2 2^{2R-1}14^R }{ |\mathcal{S}_1|\sqrt{\pi}}\frac{\log^{R-1}P}{P^R} \int_0^{\infty} \frac{z^{R-\frac{1}{2}}(1+z^{2R})}{1+2z}  e^{-z(\taa +2R)}  \mathrm{d}z.
\end{multline}
The first integral in (\ref{aljsdhlasd}) can be lower bounded by
\begin{align}
\int_0^{\infty} \frac{z^{R-\frac{1}{2}}}{1+2z}  e^{-z(\taa +2R)} \mathrm{d}z & \geq \int_0^{1} \frac{z^{R-\frac{1}{2}}}{1+2z}  e^{-z(\taa +2R)} \mathrm{d}z
 \\ & \geq \frac{e^{-(\taa +2R)}}{3}\int_0^{1} z^{R-\frac{1}{2}}\mathrm{d}z \\ \label{connncndjncd1} & = \frac{2e^{-(\taa +2R)}}{3(2R+1)}.
\end{align}
For the second integral in (\ref{aljsdhlasd}), we have
\begin{align}
\int_0^{\infty} \frac{z^{R-\frac{1}{2}}(1+z^{2R})}{1+2z}  e^{-z(\taa +2R)}  \mathrm{d}z & \leq \int_0^{\infty} z^{R-\frac{1}{2}}(1+z^{2R})  e^{-z(\taa +2R)}  \mathrm{d}z \\ & \label{connncndjncd2} = \frac{\Gamma(R+\frac{1}{2})}{(\taa +2R)^{R+\frac{1}{2}}}+
\frac{\Gamma(3R+\frac{1}{2})}{(\taa +2R)^{3R+\frac{1}{2}}}.
\end{align}
Substituting (\ref{connncndjncd1}) and (\ref{connncndjncd2}) to (\ref{aljsdhlasd}), it follows that
\begin{align}
\mathtt{NER}^\mathtt{L}(\mathtt{Q}) \geq 2\tab P^{-R} (\log^R P- \tac \log^{R-1} P ),
\end{align}
for some constants $\infty > \tab, \tac > 0$ independent of $P$.

Finally, $\tac \log^{R-1} P \leq \frac{1}{2}\log^R P ,\,\forall P\geq\exp(2\tac)$, and thus
\begin{align}
\mathtt{NER}^\mathtt{L}(\mathtt{Q}) \geq \tab P^{-R}\log^R P,
\end{align}
for all $P \geq \exp(\max\{14, 2\tac\})$. This concludes the proof.\qed
 \section{Proof of Theorem \ref{theorem2}}
 \label{proofoftheorem2}
 We provide a proof for $K>1$. The proof for $K=1$ is very similar. Thus, we skip it for brevity.

 Let $\Omega = \max_r \Omega_r$, where $\Omega_r = \gamma^{\mathtt{L}}(\mathbf{e}_r, \mathbf{h})$, as defined in Appendix \ref{dendededen}. Then, we have
\begin{align}
  \mathtt{NER}(\mathtt{Q}_{\mathcal{C}_{\mathtt{S}}})
 &   \leq C_0 \mathtt{E}[\exp(-\Omega)] \\
\label{engerek2}  &  \leq C_0 \int_0^\infty e^{-w} \sum_{r=1}^R  f_{\Omega_r}(w) \prod_{\substack{q=1\\ q\neq r}}^R\left[F_{\Omega_q}(\omega)\right]^{R-1} \mathrm{d}\omega \\
\label{engerek3} &    \leq RC_0 \lbb^R \int_0^\infty \omega^{R-1} e^{-w} \left(1 + \omega^{-\frac{1}{2\log P}} + \omega^{1-\frac{1}{\log P}}\frac{1}{P} \right)^R \mathrm{d}\omega \\
\label{engerek4} &  \leq R3^{R-1}C_0 \lbb^R \int_0^\infty \omega^{R-1} e^{-w} \left(1 + \omega^{-\frac{R}{2\log P}} + \omega^{R-\frac{R}{\log P}}\frac{1}{P^R}\right) \mathrm{d}\omega \\
\label{engerek5} &   =  R3^{R-1}C_0 \lbb^R\left[\Gamma(R) + \Gamma\left(R - \frac{R}{2\log P}\right) + \Gamma\left(2R - \frac{R}{\log P}\right)\frac{1}{P^R}\right] \\
\label{engerek6} &   \leq  \tba \frac{\log^R P}{P^R},
\end{align}
 where $\tba = R3^{R-1}C_0\lbb^R\max\{\Gamma(2R), \Gamma(\frac{1}{2})\}$, and (\ref{engerek2}) follows from the order statistics of independent RVs. For (\ref{engerek3}) and (\ref{engerek4}), we have used Proposition \ref{besselpdfcdfboundlemma}, and H\"{o}lder's inequality, respectively. This concludes the proof.\qed
\section{Proof of Theorem \ref{theorem3}}
\label{proofoftheorem3}
Let us prove the theorem for $K>1$. The proof for $K=1$ is very similar. It is thus omitted.
    
Let $\Omega$ and $\Omega_r$, be as defined in Appendix \ref{proofoftheorem2}. We need to find an upper bound on the localization distortion. According to (\ref{uboundlocdist}), it is sufficient to calculate the CNER given $|\mathcal{R}_{l}|\geq 2$.
Note that $|\mathcal{R}_{l}|\geq 2$ if and only if there exists $r,q\in\{1,\ldots,R\},\,r\neq q$ such that $\mathcal{N}(\Omega_r) = \mathcal{N}(\Omega_q) = \mathcal{N}(\Omega)$. Depending on $\mathcal{N}(\Omega)$, we divide the calculation of $\mathtt{LD}^\mathtt{U}(\xi,N)$ to two separate parts as $\mathtt{LD}^\mathtt{U}(\xi_{\mathtt{f}},N) = \sum_{i=1}^2 \mathtt{LD}_i^\mathtt{U}(\xi,N)$.

The first part is concerned with the case $\mathcal{N}(\Omega) = 0$, or equivalently, $\Omega_r \in [0, \xi_{\mathtt{f}}),\,\forall r$. Since the decoder chooses one of the $R$ relay selection vectors, the NSNR is at least $\min_r \Omega_r$. Using Proposition \ref{besselpdfcdfboundlemma}, we have
\begin{align}
\mathtt{LD}_i^\mathtt{U}(\xi_\mathtt{f},2) & \leq C_0 \int_{0}^{\xi_{\mathtt{f}}} \cdots  \int_{0}^{\xi_{\mathtt{f}}} \exp\left(-\min_r \omega_r\right) \prod_r f_{\Omega_r}(\omega_r) \prod_r \mathrm{d}\omega_r  \\ & \leq  C_0 \lbb^R \xi_\mathtt{f}^R\psi^R(\xi_\mathtt{f}) \\
&  \leq C_0 (R\lbb)^R \frac{\log^{2R} P}{P^R} \left[1 + (R\log P)^{-\frac{1}{2\log P}} +  (R\log P)^{1-\frac{1}{\log P}}\frac{1}{P} \right] \\
& \leq C_0 (R\lbb)^R \frac{\log^{2R} P}{P^R} \left(1 + e^{\frac{1}{2e}} +  \frac{R\log P}{P} \right) \\
&  \leq \tcc \frac{\log^{2R} P}{P^R},
\end{align}
for a constant $0<\tcc<\infty$, and all $P$ sufficiently large.

For the second part, we consider the case $\mathcal{N}(\Omega) = 1\iff \exists r\in\mathcal{R},\,\Omega_r \in [\xi_{\mathtt{f}}, \infty)$. In this case, the minimum NSNR is $\xi_{\mathtt{f}}$, and we simply have $\mathtt{LD}_2^\mathtt{U}(\xi_{\mathtt{f}},2) \leq C_0 \frac{1}{P^R}$.

Combining the final upper bounds for the two parts, we have $\mathtt{LD}^\mathtt{U}(\xi_{\mathtt{f}},2) \leq \tcd \frac{\log^{2R} P}{P^R}$ for some constant $0<\tcd<\infty$. This concludes the proof.\qed
\section{Proof of Theorem \ref{theorem4}}
\label{proofoftheorem4}
We prove the theorem for $K>1$. The proof for $K=1$ is very similar and skipped for brevity.

Let $\Omega$ and $\Omega_r$, be as defined in Appendix \ref{proofoftheorem2}. Also, for simplicity of notations, let $\xi = \xi_\mathtt{v}$, $N = N_{\mathtt{v}}$, and $\Xi = (N_{\mathtt{v}}-1)\xi$. Depending on $\mathcal{N}(\Omega)$, we divide the calculation of $\mathtt{LD}^\mathtt{U}(\xi,N)$ to three separate parts as $\mathtt{LD}^\mathtt{U}(\xi,N) = \sum_{i=1}^3 \mathtt{LD}_i^\mathtt{U}(\xi,N)$.

The first part is concerned with the case where $\mathcal{N}(\Omega) = 0\iff\Omega_r \in [0, \xi),\,\forall r$. In this case, the NSNR is at least $\min_r \Omega_r$. Using Proposition \ref{besselpdfcdfboundlemma}, we have
\begin{align}
\label{dp1deriv1} \mathtt{LD}_1^\mathtt{U}(\xi,N) & \leq C_0 \int_{0}^{\xi} \cdots  \int_{0}^{\xi} \exp\left(-\min_r \omega_r\right) \prod_r f_{\Omega_r}(\omega_r) \prod_r d\omega_r  \leq C_0 \lbb^R \xi^R\psi^R(\xi).
\end{align}
Now we consider the term $\psi(\xi)$ in (\ref{dp1deriv1}). For future reference, we shall calculate an upper bound for the more general quantity given by $\psi(n\xi)$, for any $n\in\{1,\ldots,N-2\}$. We have
\begin{align}
 \psi(n\xi) &   = \frac{\log P}{P} \left[1 + (n\xi)^{-\frac{1}{2\log P}} +  (n\xi)^{1-\frac{1}{\log P}}\frac{1}{P} \right] \\
&  = \frac{\log P}{P} \left(1 + n^{-\frac{1}{2\log P}}\Lambda^{\frac{1}{2\log P}} +  n^{1-\frac{1}{\log P}}\xi \Lambda^{\frac{1}{\log P}}\frac{1}{P} \right) \\
& \leq \frac{\log P}{P} \left(1 + \Lambda^{\frac{1}{2\log P}} +  n\xi \Lambda^{\frac{1}{\log P}}\frac{1}{P} \right) \\
\label{dkasklakdsadasd} & \leq \frac{\log P}{P} \left(1 + \sqrt{e} +  e n\xi \frac{1}{P} \right),
\end{align}
where the last inequality follows from $\Lambda \leq P$. Moreover, for all $n\in\{1,\ldots,N-1\}$,
$n\xi \leq (N-1)\xi \leq \log \Lambda + R \log P - R \log\log P + \frac{1}{\Lambda} \leq \epsilon^{-1} + (R+1)\log P$. Combining with (\ref{dkasklakdsadasd}), we can argue that there is a constant $0 < \tca < \infty$ such that
\begin{align}
\label{pisipisi}
 \psi(n\xi) \leq \tca \frac{\log P}{P},
\end{align}
for all $P$ sufficiently large. Using (\ref{dp1deriv1}), it follows that $\mathtt{LD}_1^\mathtt{U}(\xi,N) \leq \frac{C_0\lbb^R\tca}{\Lambda^R} \frac{\log^{R}P}{P^R}$.

For the second part, we evaluate the cases for which $\mathcal{N}(\Omega)\in\{1,\ldots,N-2\}$. For each $n\in\{1,\ldots,N-2\}$, suppose that $i \geq 2$ of $\Omega_r$ are in the interval $[n\xi, (n+1)\xi)$, and the rest $R-i$ of them are in $[0, n\xi)$. The minimum NSNR is at least $n\xi$. Also, there are $\smash{{R \choose i}}$ possible ways to choose which $\Omega_r$ will be in $[n\xi, (n+1)\xi)$. Therefore,
\begin{multline}
  \mathtt{LD}_2^\mathtt{U}(\xi,N)
\leq  \sum_{i=2}^R \sum_{\mathcal{K}\in\mathscr{K}_{i}^R} \sum_{n = 1}^{N-2} \underbrace{\int_{0}^{n\xi}\cdots\int_{0}^{n\xi}}_{R-i\,\,\mathrm{ integrals}} \underbrace{\int_{n\xi}^{(n+1)\xi}\cdots\int_{n\xi}^{(n+1)\xi}}_{i\,\,\mathrm{ integrals}} \\  \exp(-\min_{r\in\mathcal{K}}\omega_r)\prod_{r=1}^R f_{\Omega_r}(\omega_r) \prod_{r\in\mathcal{K}}\mathrm{d}\omega_r\prod_{r'\in\mathcal{K}^c}\mathrm{d}\omega_{r'},
\end{multline}
where $\mathscr{K}_{i}^R$ is the collection of all possible $i$-combinations of the set $\{1,\ldots, R\}$ (e.g. $\mathscr{K}_{2}^3 = \{\{1,2\},\{1,3\},\{2,3\}\}$), and $\mathcal{K}^c = \{1,\ldots,R\} - \mathcal{K}$. Then, similarly, we can use Proposition \ref{besselpdfcdfboundlemma} to arrive at
\begin{align}
  \nonumber \mathtt{LD}_2^\mathtt{U}& (\xi,N) \\
& \leq C_0\lbb^R  \sum_{i=2}^R \sum_{\mathcal{K}\in\mathscr{K}_{i}^R}  \sum_{n = 1}^{N-2}e^{-n\xi} \left[n\xi\psi(n\xi)\right]^{R-i} \left[\xi\psi(\xi)+ \frac{\log P}{P^2}(n\xi)^{2-\frac{1}{\log P}}\right]^i \\
& \label{antoni}  \leq C_0\lbb^R  \sum_{i=2}^R {R\choose i} 2^{i-1} \sum_{n = 1}^{N-2}  n^{R-i} e^{-n\xi}\psi^{R-i}(n\xi) \left[\xi^R \psi^i (\xi) + \xi^{R-i}\left(1 +n^{2i}\xi^{2i}\right) \frac{\log^i P}{P^{2i}} \right] \\
& \label{antoni4}  \leq C_0(4\lbb)^R \frac{\log^{R} P}{P^{R}}\sum_{i=2}^R \left[ \left(\xi^R  + \xi^{R-i}P^{-i}\right) \sum_{n = 1}^{N-2}  n^{R-i} e^{-n\xi}  +   \frac{\xi^{R+i}}{P^i} \sum_{n = 1}^{N-2} n^{R+i} e^{-n\xi} \right],
\end{align}
where (\ref{antoni}) follows from H\"{o}lder's inequality, and the fact that $(n\xi)^{2i-\frac{i}{\log P}} \leq (1+n^{2i}\xi^{2i}),\,i\geq 1$. For (\ref{antoni4}), we have applied (\ref{pisipisi}). Now, we shall evaluate the summations with respect to $n$ in (\ref{antoni4}). The following lemma provides a useful upper bound:
\begin{lemma}
\label{riemlemma}
 Let $f$ be a non-negative real valued Riemann integrable function with $f(x)<\infty,\,\forall x\in\mathbb{R}$ that is increasing on the interval $(-\infty, b)$, and decreasing on $(b, \infty)$. Then 
 \begin{align}
 \sum_{n=n_1}^{n_2} f(n) \leq \int_{n_1}^{n_2} f(x)\mathrm{d}x+2b.
 \end{align}
\end{lemma}
\begin{proof}
Let $n_b = \lfloor b \rfloor$ be the
largest integer less than $b$. Assume that $n_1<n_b<n_2$. Then
\begin{align}
\label{riem1}
\sum_{n=n_1}^{n_b-1} f(n) = \sum_{n=n_1}^{n_b-1} \int_{n}^{n+1} f(n) \mathrm{d}x
 \leq \sum_{n=n_1}^{n_b-1} \int_{n}^{n+1} f(x) \mathrm{d}x = \int_{n_1}^{n_b}f(x)\mathrm{d}x,
\end{align}
where the inequality follows from the fact that $f$ is increasing in
the range of integration. Also,
\begin{align}
\label{riem2}
\sum_{n=n_b+2}^{n_2} f(n) =
\sum_{n=n_b+2}^{n_2} \int_{n}^{n+1} f(n) \mathrm{d}x
 \leq \sum_{n=n_b+2}^{n_2} \int_{n}^{n+1} f(x-1) \mathrm{d}x =
\int_{n_b+1}^{n_2}f(x)\mathrm{d}x,
\end{align}
where the inequality follows since $f$ is decreasing on $(n_b+1, \infty)$,
and thus for $n_b+2 \leq n\leq x \leq n+1$, $f(n)\leq f(x-1)$. Finally,
combining (\ref{riem1}) and (\ref{riem2}),
\begin{align}
 \sum_{n=n_1}^{n_2} f(n) &  = \sum_{n=n_1}^{n_b-1} f(n) + f(n_b)+f(n_b+1)+\sum_{n=n_b+2}^{n_2}f(n)\\
&  \leq \int_{n_1}^{n_b-1}f(x)\mathrm{d}x + b + b + \int_{n_b+1}^{n_2}f(x)\mathrm{d}x \\
&  \leq  \int_{n_1}^{n_2}f(x)\mathrm{d}x+2b,
\end{align}
which is the desired inequality for $n_1<n_b<n_2$.
The other cases can be proved similarly. We skip them for brevity.
\end{proof}
Note that the function $f(x) = x^i\exp(-x\xi)$ has a global maximum at $x = i/\xi$ with $f(x) = i^i \xi^{-i}\exp(-i)$. Moreover, for any $0\leq a < b < \infty$, $\int_a^b f(x)\mathrm{d}x \leq \int_0^{\infty} f(x)\mathrm{d}x =\Gamma(i+1) \xi^{-(i+1)}$. Using Lemma \ref{riemlemma}, for any $i<R$, we have
\begin{align}
\sum_{n = 1}^{N-2} n^{R-i} e^{-n\xi} & \leq \Gamma(R-i+1){\xi^{-(R-i+1)}} + 2(R-i)^{R-i}\xi^{-(R-i)} \\
  & \leq  \Gamma(R-i+1){\xi^{-(R-i+1)}} + 2(R-i)^{R-i}\xi^{-(R-i+1)} \\
& \label{eq5-1} \leq  2(R-i)^{R-i} \xi^{-(R-i+1)}.
\end{align}
where the second inequality follows from the assumption that $\xi \leq 1$.

Using (\ref{eq5-1}), (\ref{antoni4}) can be bounded as:
\begin{align}
 \mathtt{LD}_2^\mathtt{U}(\xi,N) &   \leq C_0(2\lbb)^R \frac{\log^{R} P}{P^{R}}\sum_{i=2}^R \left[ 2(R-i)^{R-i}(\xi^{i-1}\!  + \xi^{-1}P^{-i}) +  2(R+i)^{R+i} \xi^{-1}P^{-i}  \right] \\
 &  \leq C_0(2\lbb)^R \frac{\log^{R} P}{P^{R}}\sum_{i=2}^R \left[ 4(R-2)^{R-2}\xi^{i-1}+  2(2R)^{2R} \xi^{i-1}  \right] \\
 &  \leq RC_0(2\lbb)^R\left[4(R-2)^{R-2}+  2(2R)^{2R}\right] \frac{\log^{R} P}{\Lambda P^{R}}
\end{align}
where the second inequality follows from the assumption that $P \geq \xi^{-1}$.

For the last part, we consider the cases for which $\mathcal{N}(Z) = N-1$. The minimum NSNR is $(N-1)\xi= \Xi \geq \log \Lambda + R \log(\frac{P}{\log P}) $, and we have $\mathtt{LD}_3^\mathtt{U}(\xi,N) \leq C_0  e^{-\Xi}  \leq C_0\frac{\log^{R} P}{\Lambda P^{R}}$.

Combining the final upper bounds for $\mathtt{LD}_i^\mathtt{U}(\xi,N),\,i=1,2,3$, $\mathtt{LD}_i^\mathtt{U}(\xi,N) \leq \tcb \frac{\log^{R} P}{\Lambda P^{R}}$ for all $P$ sufficiently large, and a constant $0 < \tcb < \infty$ that is independent of $P$ and $\Lambda$. This proves the upper bound on the LD.

Finally, by the definition of our compressor in Section \ref{compressordef}, we have
\begin{align}
\mathtt{R}_{\ell}(\mathtt{LQ}_{\xi_\mathtt{v},N_\mathtt{v}}^{\mathtt{v}}) & =\lceil \log_2 (N^R-1) \rceil \mathtt{P}(\exists r,\,\Omega_{r\ell}<N )
\\ & \leq  \lceil R \log_2 N \rceil \sum_{r=1}^R \mathtt{P}(\Omega_{r\ell}<\Xi ) \\
&  \leq R\tca \left\{1+R\log_2\left[\Lambda\log \Lambda + R \Lambda \log\left(\frac{P}{\log P}\right)+2\right]\right\} \frac{\log P}{P}, \\
&  \leq \tce \frac{\log^2 P}{P},
\end{align}
for some constant $0<\tce<\infty$, and $P$ sufficiently large. The third inequality follows from (\ref{pisipisi}). This concludes the proof.\qed
\bibliographystyle{IEEEtran}
\bibliography{IEEEabrv,it}
\end{document}

%% file: consts.tex
\newcommand{\quantg}{\mathcal{Q}^g}
\newcommand{\cbg}{\mathcal{C}^g}

\newcommand{\paa}{C_{11}}
\newcommand{\pab}{C_{12}}
\newcommand{\lbb}{C_{13}}
\newcommand{\lbc}{C_{14}}
\newcommand{\lbd}{\alpha}
\newcommand{\lbe}{\lambda_{\mathrm{y}}}
\newcommand{\lbf}{\lambda_{\mathrm{xw}}}
\newcommand{\lbg}{\bar{\lambda}}
\newcommand{\lbh}{C_{15}}
\newcommand{\lbi}{C_{16}}
\newcommand{\lbimod}{C_{16,r,\ell,\mathbf{s},\hat{\mathbf{s}}}}
\newcommand{\pac}{C_{17}}

\newcommand{\taa}{C_{18}}
\newcommand{\tab}{C_{19}}
\newcommand{\tac}{C_{20}}

\newcommand{\tba}{C_{21}}

\newcommand{\tcc}{C_{22}}
\newcommand{\tcd}{C_{23}}
\newcommand{\tca}{C_{24}}
\newcommand{\tcb}{C_{25}}
\newcommand{\tce}{C_{26}}